\documentclass{article}

\usepackage{microtype}
\usepackage{graphicx}
\usepackage{subfigure}
\usepackage{booktabs} %

\usepackage{amsmath,amsfonts,bm}

\def\1{\bm{1}}

\DeclareMathAlphabet{\mathsfit}{\encodingdefault}{\sfdefault}{m}{sl}
\SetMathAlphabet{\mathsfit}{bold}{\encodingdefault}{\sfdefault}{bx}{n}

\DeclareMathOperator*{\argmax}{arg\,max}

\usepackage{hyperref}
\usepackage{url}

\usepackage{breakurl}
\usepackage{xcolor}
\usepackage{url}
\usepackage[T1]{fontenc}
\usepackage{booktabs} %
\usepackage{footnote}
\usepackage{graphicx}
\usepackage{amsthm}
\usepackage{wrapfig}
\usepackage{tikz}
\usepackage{collcell}
\usepackage{adjustbox}
\usepackage{multirow}
\usepackage{xstring}
\usepackage{mathtools}
\usepackage{mleftright}
\usepackage{etoolbox}
\usepackage[normalem]{ulem}
\usepackage[capitalise,noabbrev]{cleveref}
\usepackage{amsthm}
\usepackage{thm-restate}
\usepackage{readarray}
\usepackage{xspace}
\usepackage{float}

\usepackage[linesnumbered,ruled,lined,noend]{algorithm2e}
\makeatletter
\patchcmd\algocf@Vline{\vrule}{\vrule \kern-0.4pt}{}{}
\patchcmd\algocf@Vsline{\vrule}{\vrule \kern-0.4pt}{}{}

\let\cref@old@stepcounter\stepcounter
\def\stepcounter#1{%
  \cref@old@stepcounter{#1}%
  \cref@constructprefix{#1}{\cref@result}%
  \@ifundefined{cref@#1@alias}%
    {\def\@tempa{#1}}%
    {\def\@tempa{\csname cref@#1@alias\endcsname}}%
  \protected@edef\cref@currentlabel{%
    [\@tempa][\arabic{#1}][\cref@result]%
    \csname p@#1\endcsname\csname the#1\endcsname}}
\makeatother

\SetKwComment{Hline}{}{\vspace{-3mm}\textcolor{gray}{\hrule}\vspace{1mm}}
\definecolor{darkgrey}{gray}{0.3}
\definecolor{commentcolor}{gray}{0.5}
\SetKwComment{Comment}{\color{commentcolor}[$\triangleright$\ }{}
\SetCommentSty{}
\SetNlSty{}{\color{darkgrey}}{}
\SetKwProg{Fn}{function}{}{}
\SetKwProg{Subr}{subroutine}{}{}
\crefalias{AlgoLine}{line}%
\crefname{algocf}{Algorithm}{Algorithms}

\newtheorem{lemma}{Lemma}

\newtheorem*{theorem*}{Theorem}
\newtheorem{corollary}{Corollary}
\newtheorem{remark}{Remark}

\newcommand{\adam}[1]{}
\newcommand{\apjacob}[1]{}
\newcommand{\noam}[1]{}
\newcommand{\david}[1]{}
\newcommand{\gabri}[1]{}
\newcommand{\jda}[1]{}

\newcommand{\qre}{piKL-HedgeBot}

\newcommand{\bp}{IL Policy\xspace}
\newcommand{\sbot}{RMBot}

\usepackage{hyperref}

\newcommand{\xsubsubsection}[1]{\stepcounter{subsubsection}\noindent\textbf{\thesubsubsection~#1}\quad}

\usepackage[accepted]{icml2022}

\icmltitlerunning{Modeling Strong and Human-Like Gameplay with KL-Regularized Search}

\begin{document}

\twocolumn[
    \icmltitle{Modeling Strong and Human-Like Gameplay with KL-Regularized Search}

    \icmlsetsymbol{equal}{*}

    \begin{icmlauthorlist}
        \icmlauthor{Athul Paul Jacob}{equal,fair,mit}
        \icmlauthor{David J. Wu}{equal,fair}
        \icmlauthor{Gabriele Farina}{equal,cmu}
        \icmlauthor{Adam Lerer}{fair}
        \icmlauthor{Hengyuan Hu}{fair}
        \icmlauthor{Anton Bakhtin}{fair}
        \icmlauthor{Jacob Andreas}{mit}
        \icmlauthor{Noam Brown}{fair}
    \end{icmlauthorlist}

    \icmlaffiliation{fair}{Meta AI Research, New York, NY, USA}
    \icmlaffiliation{mit}{CSAIL, MIT, Cambridge, MA, USA}
    \icmlaffiliation{cmu}{School of Computer Science, Carnegie Mellon University, Pittsburgh, PA, USA}

    \icmlcorrespondingauthor{Athul Paul Jacob}{apjacob@mit.edu}
    \icmlcorrespondingauthor{David J. Wu}{dwu@fb.com}
    \icmlcorrespondingauthor{Gabriele Farina}{gfarina@cs.cmu.edu}
    \icmlcorrespondingauthor{Noam Brown}{noambrown@fb.com}

    \icmlkeywords{Machine Learning, ICML}

    \vskip 0.3in
]
\printAffiliationsAndNotice{*Equal Contribution.}  %

\definecolor{goodcolor}{RGB}{200,200,255}
\definecolor{midcolor}{RGB}{240,240,240}
\definecolor{badcolor}{RGB}{255,200,200}

\newcounter{MinX}
\setcounter{MinX}{0}
\newcounter{MidX}
\setcounter{MidX}{5}
\newcounter{MaxX}
\setcounter{MaxX}{10}

\newcommand{\ApplyPositiveGradient}[1]{%
    \IfDecimal{#1}{%
        \pgfmathsetmacro{\Input}{#1 * 100}%
        \ifdim \Input pt > \value{MidX} pt%
            \pgfmathsetmacro{\PercentColor}{max(min(100.0*pow((\Input - \value{MidX})/(\value{MaxX}-\value{MidX}),0.75),100.0),0.00)}%
            \colorbox{goodcolor!\PercentColor!midcolor}{#1}%
        \else%
            \pgfmathsetmacro{\PercentColor}{max(min(100.0*pow((\value{MidX} - \Input)/(\value{MidX}-\value{MinX}),0.75),100.0),0.00)}%
            \colorbox{badcolor!\PercentColor!midcolor}{#1}%
        \fi%
    }{\bf #1}%
}
\newcommand{\ApplyNegativeGradient}[1]{%
    \IfDecimal{#1}{%
        \pgfmathsetmacro{\Input}{#1 * 100}%
        \ifdim \Input pt > \value{MidX} pt%
            \pgfmathsetmacro{\PercentColor}{max(min(100.0*pow((\Input - \value{MidX})/(\value{MaxX}-\value{MidX}),0.75),100.0),0.00)}
            \colorbox{badcolor!\PercentColor!midcolor}{#1}%
        \else%
            \pgfmathsetmacro{\PercentColor}{max(min(100.0*pow((\value{MidX} - \Input)/(\value{MidX}-\value{MinX}),0.75),100.0),0.00)}%
            \colorbox{goodcolor!\PercentColor!midcolor}{#1}%
        \fi%
    }{\bf #1}%
}
\newcolumntype{R}{>{\collectcell\ApplyPositiveGradient}r<{\endcollectcell}}
\newcolumntype{Q}{>{\collectcell\ApplyNegativeGradient}r<{\endcollectcell}}

\newcommand{\KL}[2]{D_{\mathrm{KL}}(#1\,\|\,#2)}

\begin{abstract}
\vspace{-0.03in}

We consider the task of building strong but human-like policies in multi-agent decision-making problems, given examples of human behavior. Imitation learning is effective at predicting human actions but may not match the strength of expert humans, while self-play learning and search techniques (e.g. AlphaZero) lead to strong performance but may produce policies that are difficult for humans to understand and coordinate with. We show in chess and Go that regularizing search based on the KL divergence from an imitation-learned policy results in higher human prediction accuracy and stronger performance than imitation learning alone. %
We then introduce a novel regret minimization algorithm that is regularized based on the KL divergence from an imitation-learned policy, and show that using this algorithm for search in no-press Diplomacy yields a policy that matches the human prediction accuracy of imitation learning while being substantially stronger. 

\end{abstract}

\vspace{-0.1in}
\section{Introduction}
\vspace{-0.05in}

Self-play AI algorithms have matched or exceeded expert human performance in many games, such as chess~\citep{campbell2002deep,silver2018general}, Go~\citep{silver2016mastering,silver2017mastering}, and poker~\citep{moravvcik2017deepstack,brown2017superhuman,brown2019superhuman}. However, the resulting policies often differ markedly from how humans play~\citep{mcilroy2020aligning}.
This is a serious problem for human-computer interactions that involve cooperation rather than purely competition. In such settings, modeling the other participants accurately is important for success. For example, it is important for a self-driving car at a four-way stop sign to conform to existing human conventions rather than its own self-play solution to the problem~\citep{lerer2019learning}. Moreover, even in purely adversarial games, the alien nature of AI policies makes it difficult for humans to understand and learn from superhuman bots.

The classic approach toward modeling human behavior is imitation learning (IL) on human data. However, evidence in multiple games indicates that IL on expert human data produces policies that are much weaker than actual expert human play in domains with complex strategic planning.
In this paper, we study the problem of producing policies that are both \emph{strong} and \emph{human-like} 
in games with complex strategic planning like chess, Go, Hanabi, and Diplomacy. In all four, we find that conducting search with KL-regularization towards an IL policy matches or exceeds the prior state of the art for prediction accuracy of expert humans while also better matching expert human performance.

In Section~\ref{sec:mcts}, we show that Monte Carlo tree search (MCTS) with a human imitation-learned policy prior and value function surpasses prior state-of-the-art results for human prediction accuracy in chess and Go. As explained by \citet{grill2020monte}, standard MCTS with a policy prior can be viewed as a form of KL-regularized search, optimizing a value function subject to a KL-divergence term with that prior. Although MCTS has been extensively studied for developing strong agents, it has been explored much less in the context of developing human-like agents. %

Section~\ref{sec:rm} builds on these findings and shows how to generalize them to a class of imperfect-information games 
(in which ordinary MCTS is unsound and cannot be applied) via a new algorithm for KL-regularized regret minimization.
We show that existing regret minimization algorithms achieve low accuracy in predicting expert human actions in no-press Diplomacy. We then introduce the first regret minimization algorithm to incorporate a cost term proportional to the KL divergence between the search policy and a human-imitation learned \textbf{anchor policy}. We call this algorithm \textbf{policy-regularized hedge}, or \textbf{piKL-hedge}. We prove that piKL-hedge converges to an equilibrium in which all players' policies are optimal given the joint policies of the players and the cost of deviating from the anchor policy. We then present results in no-press Diplomacy showing that piKL-hedge produces policies that predict human actions as accurately as imitation learning while also improving head-to-head performance in a population of prior agents.

Appendix~\ref{appendix:hanabi} additionally shows that applying KL-regularization toward a human IL policy in the search algorithm SPARTA~\cite{lerer2020improving} produces similar or better human prediction accuracy while greatly improving performance in the benchmark domain of Hanabi~\cite{bard2020hanabi}.

Our experiments demonstrate the benefits of KL-regularized search in all four of chess, Go, no-press Diplomacy, and Hanabi to producing agents that are simultaneously more human-like and closer in strength to actual human experts than purely imitation-learned models. %

\vspace{-0.05in}
\section{Preliminaries}
\vspace{-0.05in}
We study the problem of learning policies for multiplayer games. Here we briefly introduce the key ingredients of both classes of games we study; \cref{sec:mcts} and \cref{sec:rm} give a more formal presentation tailored to individual game types and learning algorithms.

An ($N$-player) game is defined by a \textbf{state space} $S$, an \textbf{action space} $A$, a (deterministic) \textbf{transition function} $T : S \times A^N \to S$, and a collection of \textbf{reward functions} $u_i$. We model the behavior of each player in a game as a \textbf{policy} $\pi_i : S \to \Delta(A)$ (a distribution over actions given states). In every round of a game, each player observes a (possibly incomplete) view $s_i^t$ of the current state. One or more players select actions $a_i^t \sim \pi_i(\cdot \mid s_i^t)$,
then each player receives a reward $u_i^t(s^t, \mathbf{a}^t = a_1^t, \ldots, a_n^t)$, and the game transitions into a new state $s^{t+1} = T(s^t, \mathbf{a}^t)$. 
Each player $i$ aims to maximizes its expected reward, and the optimal policies for doing so may depend on the policies $\pi_{-i} = \{ \pi_1, \ldots, \pi_{i-1}, \pi_{i+1}, \ldots, \pi_N \}$ of the other players.
 
The sequential decision-making problem described above is extremely general, and in this paper we focus on two special cases. In \textbf{perfect-information} games, players make moves sequentially (e.g.,\ $u^1$ and $s^2$ depend only on $a_1$, $u^2$ and $s^3$ depend only on $a^2$, etc.). Many important games, including chess and Go, fall into this category. Next, we study a more general class of \textbf{imperfect-information, simultaneous-action} games that make no assumptions about the dependence of different $u_i$ and $T$ on $\mathbf{a}$; here we focus on games with only a single round, also called matrix games.
Owing to the large differences between these two settings, the tools for computing strong policies are quite different. The remainder of this paper accordingly offers a deeper exploration of each class of games: perfect-information games in \cref{sec:mcts} and imperfect-information games in \cref{sec:rm}.

\jda{the above can probably be cleaned up / made more consistent with the existing literature / our notation in this paper. the treatment of imperfect information in particular is currently sloppy. but I think this is roughly what we're looking for}

\vspace{-0.05in}
\section{Perfect-Information Games: Policy Regularization in Monte Carlo Tree Search}
\vspace{-0.05in}
\label{sec:mcts}

In this section, we focus on developing strong human-like agents for perfect-information games. Monte Carlo tree search (MCTS) has been highly successful for developing strong, but not necessarily human-like agents in this setting, and is a key component of general learning algorithms such as AlphaZero and MuZero capable of achieving superhuman performance in chess, Go, and similar games \citep{silver2018general,schrittwieser2020mastering}. By contrast, for developing human-like agents, the best prior human move prediction accuracies for chess and Go were all achieved with pure imitation learning on human data \cite{mcilroy2020aligning,Cazenave2017,silver2017mastering}.

The state of the art for predicting human moves in chess is the Maia engine created by \citet{mcilroy2020aligning} via pure imitation learning without any search.
However, this approach appears to be of limited effectiveness for modeling sufficiently strong humans. Although the weakest Maia models at low temperatures appear to outperform the players they imitate (due to ``averaging away'' of the imitated players' individual idiosyncratic mistakes~\citep{MaiaGuestPost}) each successive model on data from stronger players improves by much less than the players improve.\footnote{See ratings data at https://lichess.org/@/maia1, https://lichess.org/@/maia5, https://lichess.org/@/maia9} The strongest model, trained to predict human 1900-1999 rated players, even with low temperature appears to be clearly below a 1950-average level of performance in all but the minority of bullet-speed games (in which very little time is available for planning and players are forced to rely more heavily on intuition).
Similarly, in Go, pure imitation-learning agents have not exceeded mid-expert level on various online servers despite being trained on top-expert and master-level games~\citep{Cazenave2017}.

In contrast, search-based reinforcement learning (RL) agents such as AlphaZero that do not use a human policy prior play at a superhuman level, but often in non-human ways that humans find difficult to understand even when given access to interactively query and inspect the agent's analysis~\citep{silver2017mastering,philosophies5040037}.

However, we show in both chess and Go that if the human-learned model is used in MCTS with appropriate parameters, MCTS outperforms those models' human prediction accuracy while simultaneously reducing the shortcomings in those models' strength.

\vspace{-0.05in}
\subsection{Background}
\vspace{-0.05in}
We consider sequential games where each player~$i$ alternatively chooses action $a$ from a policy $\pi_i$ where, $a \sim \pi_i(\cdot\mid s)$. Each action deterministically transitions the game into a new state $s' = T(s, a)$ and gives rewards $u_i(s, a)$. Notationally, we may elide the player $i$ in some places when it is clear that $i$ is the next player to move in the state being considered.

For this work, following \citet{silver2016mastering}, we implement one of the most common modern forms of MCTS, which uses a value function $V$ predicting the expected total future reward
$V_i(s) = \mathbb{E} [\sum_t u_i(s_t,a_t) \mid \pi_1, \pi_2, s_0=s]$ and a policy prior $\tau$ (typically both the outputs of a trained deep neural net) and attempts to produce an improved policy $\pi$.

Each turn, MCTS builds a game tree over multiple iterations rooted at the current state. Each iteration $t$, MCTS explores a single path down the tree by simulating at each successive state $s$ with player $i$ to move the action:
\begin{equation}
    \argmax_a \, Q(s,a) + c_\text{puct}\tau(s,a)\frac{\sqrt{\sum_b N(s,b)}}{N(s,a)+1}
\end{equation}
where $Q(s,a)$ is the estimated expected future reward for $i$ from playing action $a$ in state $s$, the visit count $N(s,a)$ is the number of times $a$ has been explored from $s$, $\tau(s,a)$ is the prior policy probability, and $c_\text{puct}$ is a tunable parameter trading off exploration versus exploitation.

Upon reaching a state $s_t$ not yet seen, MCTS queries the value function $V_i(s_t)$ for each player $i$, and based on $V_i(s_t)$ and any intermediate rewards received, updates all $Q(s,a)$ estimates on the path traversed. The final agent policy is $\pi(s,a) = N(s,a)/\sum_b N(s,b)$ where $s$ is the root state, or optionally we may also have $\pi(s,a) \sim N(s,a)^{1/T}$ where $T$ is a temperature parameter. See also Appendix \ref{appendix:mcts} for a fuller description of MCTS.

\citet{grill2020monte} show that the agent policy $\pi$ computed by this form of MCTS is an approximate solution to the optimization problem:
\begin{equation}
    \argmax_{\pi} \sum_a Q(s,a)\pi(s,a) + \lambda \KL{\tau}{\pi}
\end{equation}
where $\lambda \sim c_\text{puct}\sqrt{N}$ and $N$ is the total number of iterations.

In other words, at every node of the tree recursively, MCTS implicitly optimizes its expected future reward subject to KL regularization of its policy towards the prior policy $\tau$ with strength controlled by $\lambda$. For any fixed computational budget $N$, we can therefore tune $c_\text{puct}$ to vary the strength of that prior, with $c_\text{puct} = \infty$ approximating the prior policy before search, and $c_\text{puct} \rightarrow 0$ approaching a greedy argmax of the $Q$ value estimates.

If our goal is a strong \emph{human-like} agent rather than solely a strong agent, and the KL-regularizing policy is learned from human data, then that policy serves not just as a prior, but also as an \textbf{anchor policy} that regularizing towards is desirable in and of itself. With good choice of $c_\text{puct}$, MCTS can improve that policy while remaining close to human. Our experiments confirm that MCTS improves on the strength and human prediction accuracy of the best existing models in both chess and Go.

\vspace{-0.05in}
\subsection{Experiments in Chess and Go}
\vspace{-0.05in}
\label{sec:experimentschessgo}

\label{sec:improvedpredictionchess}

\begin{figure*}
    \begin{minipage}{9.4cm}
        \begin{table}[H]
            \scalebox{.9}{%
                \setlength{\tabcolsep}{.8mm}
                \def\arraystretch{1.31}
                \small
                \begin{tabular}{@{}rrrrRRRRR@{}}
                    \toprule
                    \multirow{2}{*}{\bf Game} & \multirow{2}{*}{\textbf{Model}} & \multirow{2}{*}{\textbf{ Predicting}} & \multirow{2}{*}{\begin{minipage}{0.9cm}~~\textbf{Raw}\newline\textbf{Model}\end{minipage}} & \multicolumn{5}{c}{\textbf{Model + MCTS}}                                   \\ \cmidrule(l){5-9}
                                              &                                 &                                       &                                            & $c_\text{puct}\!=\!10$                    & ${5}$ & ${2}$ & ${1}$ & ${0.5}$ \\
                    \midrule
                    \setcounter{MinX}{4750}%
                    \setcounter{MidX}{5110}%
                    \setcounter{MaxX}{5200}%
                    Chess                     & Maia1100                        & Rating 1100                           & \ApplyPositiveGradient{51.1}               & 51.2                                      & 51.3  & 50.8  & 49.5  & 47.4    \\
                    \setcounter{MinX}{5000}%
                    \setcounter{MidX}{5240}%
                    \setcounter{MaxX}{5340}%
                    Chess                     & Maia1500                        & Rating 1500                           & \ApplyPositiveGradient{52.4}               & 52.7                                      & 52.9  & 52.8  & 51.9  & 50.1    \\
                    \setcounter{MinX}{5240}%
                    \setcounter{MidX}{5320}%
                    \setcounter{MaxX}{5440}%
                    Chess                     & Maia1900                        & Rating 1900                           & \ApplyPositiveGradient{53.2}               & 53.6                                      & 54.0  & 54.3  & 53.8  & 52.4    \\
                    \midrule
                    \setcounter{MinX}{5470}
                    \setcounter{MidX}{5780}
                    \setcounter{MaxX}{5850}
                    Go                        & \citeauthor{Cazenave2017}       & GoGoD                                 & \ApplyPositiveGradient{54.7}                                                                                             \\
                    Go                        & \citet*{wu2018gonn}             & GoGoD                                 & \ApplyPositiveGradient{55.3}                                                                                             \\
                    Go                        & \citeauthor{Wang2017DANNLTE}             & GoGoD                                 & \ApplyPositiveGradient{57.7}                                                                                             \\
                    Go                        & Ours                            & GoGoD                                 & \ApplyPositiveGradient{57.8}               & 58.1                                      & 58.3  & 58.5  & 58.1  & 57.1    \\ \bottomrule
                \end{tabular}}
            \vspace{-0.1in}
            \caption{\% top-1 test accuracy predicting human moves in chess and Go using MCTS with 50 playouts and various $c_\text{puct}$, or raw model without MCTS, equivalent to infinite $c_\text{puct}$. In chess, first 10 ply per game and moves with < 30s time left excluded, similar to \cite{mcilroy2020aligning}. Standard error is approx 0.1 or less on all values. MCTS on top of current state-of-the-art models improves human prediction accuracy significantly.
                \label{table:chess-and-go-accuracy}
            }
        \end{table}
    \end{minipage}%
    \hfill%
    \begin{minipage}{7.4cm}
        \begin{figure}[H]%
            \centering%
            \begin{tikzpicture}%
                \tikzstyle{lbl}=[rounded corners=.3mm,draw=gray!40,fill=white,inner ysep=.3mm]
                \node at (0, 0) {\includegraphics[scale=.81]{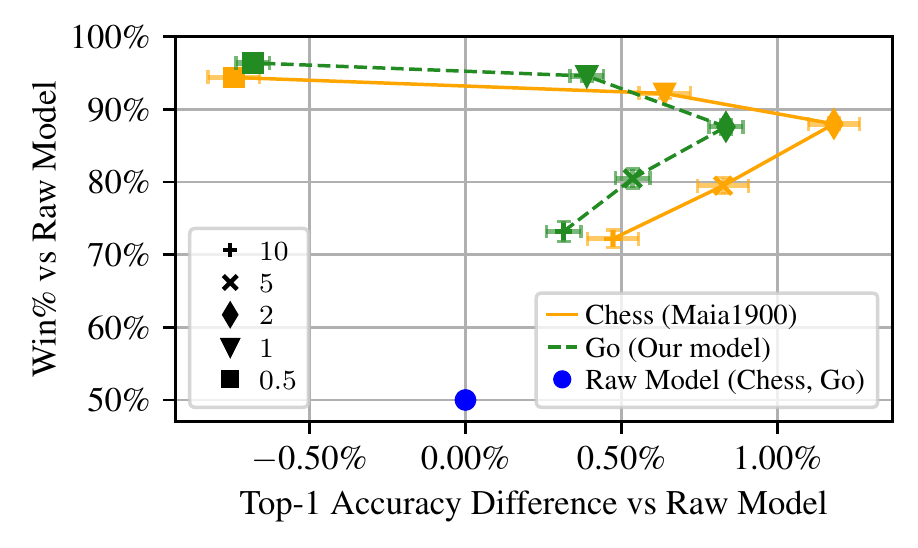}};
                \node[lbl] at (-17.5mm, 4.5mm) {\scalebox{.8}{\bf $c_\text{puct}$}};
                \node[lbl] at (20mm, -1mm) {\scalebox{.6}{\bf Model}};
            \end{tikzpicture}\\[-5mm]
            \caption{\small Improvement in top-1 accuracy of Maia1900 in Chess or our GoGoD model in Go using MCTS, plotted versus winrate of MCTS against the raw model (temperature 1). Error bars indicate $1$ standard error. Many $c_\text{puct}$ values increase both human prediction accuracy and winrate over the raw model in both Chess and Go.}
            \label{fig:gochesscurve}
        \end{figure}
    \end{minipage}
    \vspace{-0.1in}
\end{figure*}

In chess and Go, we ran two main experiments each. First, in chess using the prior state-of-the-art Maia models from \citet{mcilroy2020aligning} and in Go using a model trained on professional games from the GoGoD dataset, we demonstrate that MCTS with that model outperforms the raw model in human prediction accuracy. Second, we also sanity-check that MCTS with the same parameters greatly improves the strength of the same models in chess and Go. %

\xsubsubsection{Data and Model Architecture}
\label{sec:chessgoarchitecture}
In chess, for the human-learned anchor policy we use the pre-trained Maia1100, Maia1500, and Maia1900 models from \citet{mcilroy2020aligning}, achieving state-of-the-art performance on rating-conditional human move prediction. These models follow a standard AlphaZero-like residual block architecture, including both a policy and a value head, and were trained to imitate players in ratings ``buckets'' 1100, 1500, and 1900 respectively, based on roughly 10 million games each from the popular Lichess server (each bucket contains games between players from rating N to N+99). %
\begin{figure}[ht!]
    \centering
    \includegraphics[width=8cm]{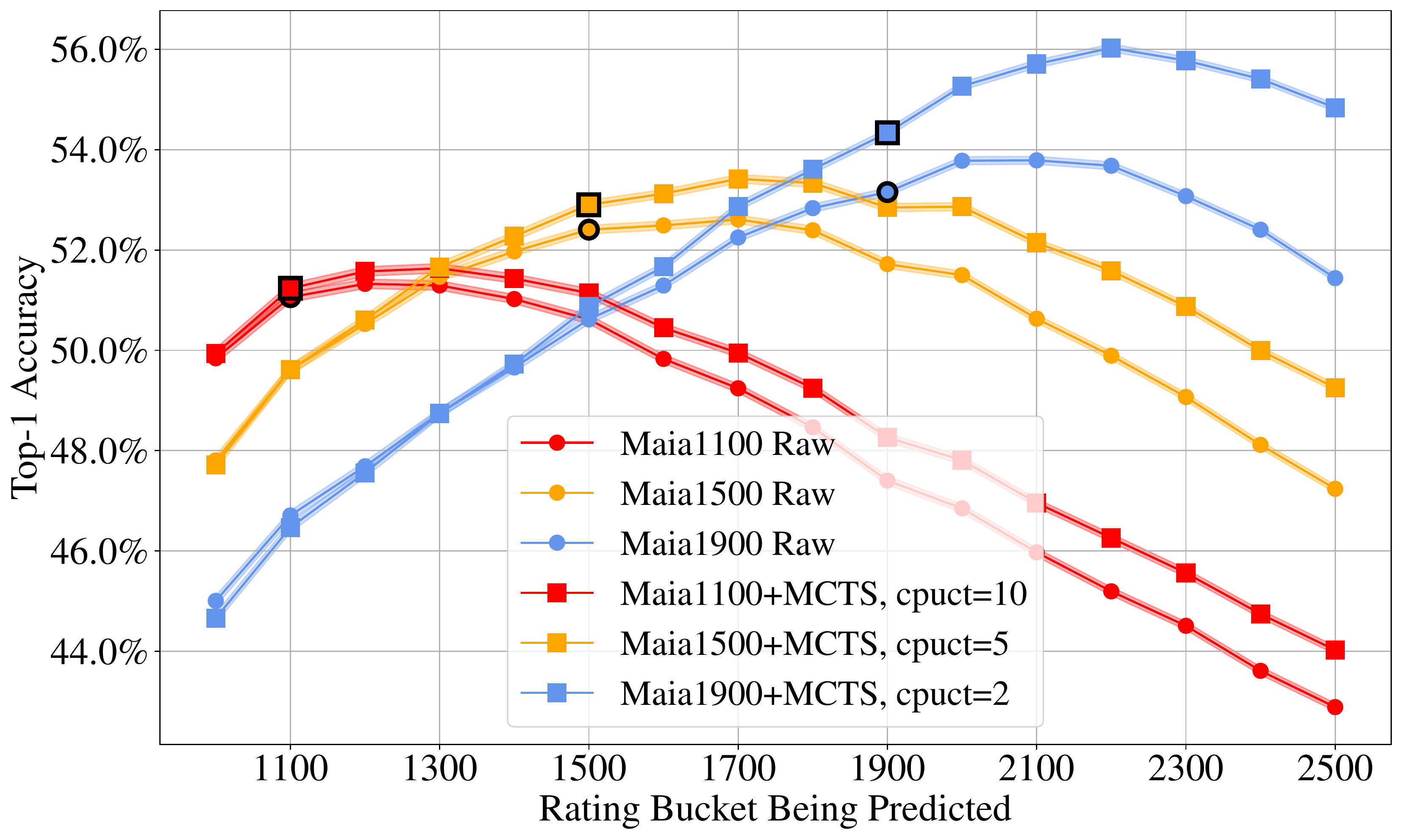}
    \vspace{-0.2in}
    \caption{Top-1 \% test accuracy in chess for Maia models trained to predict moves by players in rating buckets 1100, 1500, 1900, applied to predict \emph{all} rating buckets, with and without MCTS. The black bolded outline indicates where the model is predicting the same rating as it was trained on. Standard error is very small, the slight thin shade around each line. MCTS most improves prediction of a model on players of its target rating and higher.}
    \label{fig:maiabyrating}
    \vspace{-0.1in}
\end{figure}

For Go, we trained a deep neural net on the GoGoD professional game dataset\footnote{https://gogodonline.co.uk/}. We match \citet{Cazenave2017} in using games from 1900 through 2014 for training and 2015-2016 as the test set, with roughly 73000 and 6500 games, respectively.
Our architecture matches the 20-block residual net structure used by some versions of AlphaZero \citep{silver2017mastering}, except adds squeeze-and-excitation layers which have been successful in image processing tasks~\citep{hu2018SE} and self-play learning in chess and Go~\citep{lc0net,minigoSE}. See Appendix \ref{appendix:gotraining} for additional details.

\xsubsubsection{Improved Human Prediction and Strength}
In Table \ref{table:chess-and-go-accuracy} we show the top-1 accuracy of MCTS at predicting players in chess and Go. MCTS on top of each model tested outperforms that model at predicting human moves.

In chess, the benefit provided increases greatly as the rating of players predicted increases, while the optimal choice for $c_\text{puct}$ appears to decrease (allowing increasingly small value differences to affect the search). This is consistent with the intuitive hypothesis that stronger players plan more deeply, increasing the value of explicitly modeling planning, and that they are more sensitive to small future value differences. In Go, despite our baseline model being equal or better than all prior imitation-learning models on the GoGoD human pro games dataset, MCTS improves it yet further.

In Figure \ref{fig:maiabyrating}, for chess we see that while KL-regularized search improves each model's accuracy on players of its target rating, surprisingly, the improvement grows yet larger when each model predicts players of \emph{higher} rating than it was trained on. This suggests that as human players improve, the incremental average change in their behavior resembles or is correlated with the way that highly-regularized search improves the strength of a baseline policy.

Additionally, in Appendix \ref{appendix:piklcrossentropy}, we show that if we apply post-processing to the MCTS policy based on \citet{grill2020monte}, MCTS improves cross entropy with human moves in both chess and Go, not just top-1 accuracy. In other words, not only does policy-regularized search improve the prediction of the top move, but it also better models the overall \emph{distribution} of moves that humans may likely play.

We measure the strength impact of regularized search with 1000 games\footnote{In Go, we also use the open-source KataGo~\citep{Wu2020Go} to determine when the game is over and to score the result. Unlike RL agents, humans which our models imitate universally pass and score the game well before it becomes mechanically scorable, so we use KataGo as a neutral judge.} per $c_\text{puct}$ setting between the raw model policy and the MCTS policy, sampling each at temperature 1.
Figure \ref{fig:gochesscurve} shows the change in human prediction accuracy of MCTS in both chess and Go plotted jointly versus winrate of MCTS against the raw model. Rather than solely a tradeoff between strength and accuracy, most $c_\text{puct}$ values in the range we tested increase \emph{both}, some achieving more than 90\% winrate while still improving human prediction. See Appendix \ref{appendix:chessgo_experiments} for results at lower temperature and evidence that accuracy improves further at longer time controls. 

Although we did not test against humans directly to calibrate, this gives clear evidence that a well-tuned human-regularized MCTS agent would be better able to match the 1900-1999-rated chess players that Maia1900 currently falls hundreds of Elo short of imitating, while simultaneously being more accurate to their move-by-move behavior, and similarly for human-imitation agents in Go.

\vspace{-0.05in}
\section{Imperfect-Information Games: Policy-regularized Regret Minimization}
\vspace{-0.05in}
\label{sec:rm}

While MCTS is a popular search algorithm for perfect-information deterministic games, it is not able to compute optimal policies in imperfect-information games.
Instead, iterative algorithms based on regret minimization are the leading approach to search in imperfect-information games.

\textbf{Hedge}~\citep{littlestone1994weighted,freund1997decision} is an iterative regret minimization algorithm that in general converges to a coarse correlated equilibrium~(CCE)~\citep{hannan1957approximation}. In the special case of two-player zero-sum games, it also converges to a Nash equilibrium~(NE)~\citep{nash1951non}. 

\textbf{Regret Matching~(RM)}~\citep{blackwell1956analog,hart2000simple} is another equilibrium-finding algorithm similar to Hedge that has historically been more popular and that we compare our algorithm to in this paper.

The effectiveness of the implicit KL-regularization in MCTS that we study in \cref{sec:mcts} motivates us to develop an equilibrium-finding algorithm called piKL-Hedge that similarly biases the search towards an anchor policy. In \cref{sec:results_diplomacy}, we show that piKL-Hedge achieves better empirical performance against baseline human-imitation models than Hedge and RM in a large imperfect-information game, as well as much higher human prediction accuracy.

\renewcommand{\vec}[1]{\bm{#1}}
\newcommand{\anch}{\pi_{\text{anchor}}}
\newcommand{\defeq}{\coloneqq}

\vspace{-0.05in}
\subsection{Background}
\vspace{-0.05in}
\label{sec:hedge_background}
We consider a game with $\mathcal{N}$ players where each player~$i$ chooses an action~$a$ from a set of actions $\mathcal{A}_i$. We denote the actions of all players other than~$i$ as $\vec{a}_{-i}$. After all players simultaneously choose an action, player~$i$ receives a reward of $u_i(a, \vec{a}_{-i})$. Players may also choose a probability distribution over actions, where the probability of action~$a$ is denoted $\vec{\pi}_i(a)$ and the vector of probabilities is denoted $\vec{\pi}_i$. We also define the fixed policy that we are interested in biasing player $i$ towards as the anchor policy $\vec{\tau}_i \in \Delta(A_i)$.

Each player~$i$ maintains a \textbf{regret} for each action. The regret on iteration~$t$ is denoted $R_i^t(a)$. Initially, all regrets are zero. On each iteration~$t$ of Hedge, $\vec{\pi}^t_i(a)$ is set according to
\begin{equation} \label{eq:hedge}
    \vec{\pi}^{t}_i(a) \propto
        \exp\big(R^t_i(a)\big)
\end{equation}

Next, each player samples an action $a^*$ from $\mathcal{A}_i$ according to $\pi_i^t$ and all regrets are updated such that
\begin{equation} \label{eq:srm}
R_i^{t+1}(a) = R_i^{t}(a) + u_i(a, \vec{a}^*_{-i}) - \sum_{a' \in \mathcal{A}_i} \vec{\pi}^t_i(a') u_i(a'_i, \vec{a}^*_{-i})
\end{equation}
It is proven that the \emph{average} policy of Hedge over all iterations converges to a NE in two-player zero-sum games and, more broadly, the players' joint policy distribution converges to a CCE as $t \rightarrow \infty$.

We wish to model agents that seek to maximize their expected reward, while at the same time playing ``close'' to the anchor policy. The two goals can be reconciled by defining a composite utility function that adds a penalty based on the ``distance'' between the player policy and their anchor policy, with coefficient $\lambda_i\in[0,\infty)$ scaling the penalty. %

For each player~$i$, we define $i$'s utility as a function of the the agent policy $\vec{\pi}_i \in \Delta(A_i)$ given policies $\vec{\pi}_{-i}$ of all other agents:
\begin{align}\label{eq:regularized C}
\mathcal{U}_{i}(\vec{\pi}_i, \vec{\pi}_{-i}) &\defeq u_i(\vec{\pi}_i, \vec{\pi}_{-i}) - \lambda_i\,\KL{\vec{\pi}_i}{\vec{\tau}_i}. 
\end{align}

When $\lambda$ is large, the utility function is dominated by the KL-divergence term $\lambda_i\,\KL{\vec{\pi}_i}{\vec{\tau}_i}$, and so the agent will naturally tend to play a policy $\vec{\pi}_i$ close to the anchor policy $\vec{\tau}_i$ \footnote{The careful reader may observe that the direction of the KL-divergence term, $\KL{\pi}{\tau}$ is the opposite of the direction implicit in MCTS, $\KL{\tau}{\pi}$. We choose this direction for greater ease of theoretical analysis and implementation in the context of regret minimization; for our use cases we have not found the exact form of the loss to be critical so much as simply doing any reasonable regularized search.}. When $\lambda_i$ is small, the dominating term is the rewards $u_i(\vec{\pi}_i, \vec{a}^t_{-i})$ and so the agent will tend to maximize reward without as closely matching the anchor policy $\vec{\tau}_i$. These statements are made precise in ~\cref{cor:distance} and \cref{cor:exploitability}.

\vspace{-0.05in}
\subsection{No-Regret Learning for Policy-Regularized Utilities}
\vspace{-0.05in}
In this section, we present \cref{algo:noregret}, a no-regret algorithm based on Hedge for any player~$i$ to learn strong policies relative to the regularized utilities defined in~\eqref{eq:regularized C}. 
As we show in \cref{prop:regret} in \cref{sec:proofs}, it guarantees that each player $i$ accumulates sublinear regret (of order $\log T$) with respect to the regularized utility functions:
\[
    \mathcal{U}^t_{i}(\vec{\pi}_i) \defeq \mathcal{U}_{i}(\vec{\pi}_i, \vec{a}_{-i}^t) = u_i(\vec{\pi}_i, \vec{a}^t_{-i}) - \lambda_i\,\KL{\vec{\pi}_i}{\vec{\tau}_i},
\]
no matter the opponents' actions $\vec{a}_{-i}^t$ at each time $t$.

\begin{figure}[ht]
\vspace{-0.05in}
    \begin{minipage}{\columnwidth}
        \SetInd{0.4em}{0.6em}
        \begin{algorithm}[H]\caption{\textsc{piKL-Hedge} (for Player~$i$)}\label{algo:noregret}
            \DontPrintSemicolon
            \KwData{\mbox{~\textbullet~} $A_i$ set of actions for Player~$i$;\newline
                \mbox{~\textbullet~} $u_i$ reward function for Player~$i$;\newline
                \mbox{~\textbullet~} $\eta > 0$ learning rate hyperparameter.}
            \BlankLine
            \Fn{\normalfont\textsc{Initialize}()}{
                $t \gets 0$\;

                \For{\normalfont\textbf{each} action $a \in A_i$}{
                    $\mathsf{CV}^0_i(a) \gets 0$\;
                }
            }
            \Hline{}
            \Fn{\normalfont\textsc{Play}()}{
                $t \gets t + 1$\;
                
                let $\vec{\pi}^t_{i}$ be the policy such that 
                \begin{equation}\label{eq:distribution}
                    \vec{\pi}^t_{i}(a) \propto \exp\mleft\{\frac{\eta\,\mathsf{CV}^{t-1}_i(a) + t\lambda_i\eta\,\log\tau_i(a)}{1+t\lambda_i\eta}\mright\}.
                \end{equation}\;\label{line:pick a}\vspace{-4mm}

                sample an action $a^t \sim \vec{\pi}^t_{i}$\;

                play $a^t$ and observe actions $\vec{a}^t_{-i}$ played by the opponents\;
                
                \For{\normalfont\textbf{each} $a \in A_i$}{
                    $\mathsf{CV}^{t}_i(a) \gets \mathsf{CV}^{t-1}_i(a) + u_i(a, \vec{a}^t_{-i})$\;
                }
            }
        \end{algorithm}
\end{minipage}
\end{figure}

As with many other regret-minimization methods, we consider the \emph{average} policy of each player $i$ over $T$ iterations: %
\begin{equation}\label{eq:avg policy}
    \vspace{-0.03in}
    \bar{\vec{\pi}}_{i}^T \defeq \frac{1}{T}\sum_{t=1}^T \vec{\pi}^t_{i}
    \vspace{-0.03in}
\end{equation}
where $\vec{\pi}^t_{i}$ is defined in~\eqref{eq:distribution}. We take $\bar{\vec{\pi}}^T_i$ to be the final agent policy produced by \qre in no-press Diplomacy (as described in more detail in \cref{sec:results_diplomacy}). %
As shown in \cref{sec:proofs}, the KL-divergence of $\bar{\vec{\pi}}^T_i$ from the anchor policy $\vec{\tau}_i$ converges to be inversely proportional to $\lambda_i$:

\begin{restatable}{theorem}{thmdistancefromtau}\label{cor:distance}
    (piKL stays close to the anchor policy) Upon running \cref{algo:noregret} for $T$ iterations in a multiplayer general-sum game, the policy $\bar{\vec{\pi}}_i^T$ is at a distance
    \[
        \KL{\bar{\vec{\pi}}^T_i}{\vec{\tau}_i} \le \frac{1}{\lambda_i}\mleft(\frac{R_i^T}{T} + D_i\mright),
    \]
    where $D_i$ is any upper bound on possible rewards for Player $i$. 
    In particular, if $\eta > 0$ is set so that $R_i^T = o(T)$, then $\KL{\bar{\vec{\pi}}^T_i}{\vec{\tau}_i} \to D_i/\lambda_i$ as $T \to +\infty$.
\end{restatable}

We can also show (see \cref{sec:proofs}) that in the case of a \emph{two-player zero-sum} game, $\bar{\vec{\pi}}^T_i$ approximates a Nash equilibrium of the original utility functions, with the approximation guarantee controlled by $\lambda$:

\begin{restatable}{theorem}{thmexploitability}\label{cor:exploitability}
    Let $(\bar{\vec{\pi}}_1,\bar{\vec{\pi}}_2)$ be any limit point of the average policies $(\bar{\vec{\pi}}_{1}^T, \bar{\vec{\pi}}_2^T)$ of the players. Almost surely, $(\bar{\vec{\pi}}_1,\bar{\vec{\pi}}_2)$ is a $(\max_{i=1,2}\{\lambda_i\beta_i\})$-approximate Nash equilibrium policy with respect to the original utility functions $u_i$, where $\beta_i$ is as defined in \eqref{eq:def beta i}.
    \jda{seems like we should probably mention that this is vacuous if any 1/tau(a) is bigger than the greatest difference between rewards. when should we expect this not to be the case? e.g.\ with a diplomacy-sized action space it's definitely problematic...} \adam{I think it's ``if any $\lambda log(1/tau(a))$ is bigger...'', which is much less of a big deal because of the log and because $\lambda$ is typically small. Plus I think that it's really only among actions $a$ in the support of the equilibrium, although I'm not sure (Gabriele?)} 
    \jda{oh, if it's only about the support of the equilibrium we should definitely mention that}\gabri{Adam: the proof we have right now would not be strong enough to support that, and I'm not 100\% the result would hold with what you're saying. But I'm happy to try and think about it if you guys think it would make a substantial difference.}\adam{Well, I suspect you could at least restrict $\tau(a)$ to non-dominated actions? because if there's an irrelevant awful action that has $\pi(a)=0$ then the bound goes to infinity but the piKL equilibrium should be the same quality. Anyway, even leaving this aside, imagine that $1/\tau \geq 10^{-10}$, then with $\lambda=1e-2$ we still have $\lambda\log(1/\tau(a)) = 0.23.$}
\end{restatable}

Lastly, we remark that in the special case that $\vec{\tau}_i$ is the uniform policy for all players $i$, the above results imply that \cref{algo:noregret} converges towards a \textbf{quantal response equilibrium}~\citep{mckelvey19956qre}, in which an imperfect agent is modeled as choosing actions with probability exponentially decaying in the amount that each action is worse than the best action(s), given that all other agents behave the same way. Our method can be seen as a generalization that takes into account a human-learned prior for what actions may be more likely.

\vspace{-0.05in}
\subsection{Diplomacy Experiments}
\vspace{-0.05in}
\label{sec:results_diplomacy}
Diplomacy is a benchmark 7-player simultaneous-action game featuring both cooperation and competition. In \cref{appendix:rules_diplomacy}, we summarize the rules of the game.
Using piKL-Hedge, we develop an agent piKL-HedgeBot and show that it improves upon prior approaches for the game.
In \cref{appendix:blotto}, we also illustrate the key features of piKL-Hedge in Blotto, a famous 2-player simultaneous action game.

\xsubsubsection{Algorithms and Models}
\label{sec:piklhedge_results}
In no-press Diplomacy, we compare the different equilibrium search algorithms (RM, Hedge and piKL-Hedge) using the procedure introduced by \citet{gray2020human}. We perform 1-ply lookahead where on each turn, we sample up to 30 of the most likely actions for each player from a policy network trained via imitation learning on human data (IL policy). We then consider the 1-ply subgame consisting of those possible actions where the rewards for a given joint action are given by querying a value network trained on human game data as in \citet{gray2020human}. We play according to the approximate equilibrium computed for that subgame by that algorithm. For piKL-Hedge, the anchor policy is simply the same human-trained policy network. Our baseline policy and value models also contain a few improvements over prior models for no-press Diplomacy, described in Appendices \ref{appendix:diplomacyarchitecture} and \ref{appendix:propval}. 

In our experiments, we label our RM, Hedge, and piKL-Hedge agents as \sbot, HedgeBot, and piKL-HedgeBot respectively. We compare also against SearchBot \cite{gray2020human} (similar to \sbot~but using the models from \citet{gray2020human} rather than our models).
See \cref{appendix:hyperparams} for more details about the hyperparameters used.

\xsubsubsection{Strong, human-like play with piKL-Hedge}
\label{sec:result_strong_human}
Similar to Chess and Go, we compare the human prediction accuracy of RMBot, HedgeBot, piKL-HedgeBot (with different $\lambda$s) to the IL anchor policy, as well as testing their head-to-head strength. In particular, we test their ability to predict human moves in 226 no-press Diplomacy games from a validation set, and measure their score against the IL policy across 700 games each.

In \cref{fig:diplomacysosaccuracycurveL} (Left), we present the average top-1 accuracy of unit orders in each action predicted by these methods as well as their average scores against 6 IL anchor policies. The raw IL model ($\lambda = \infty$) predicts human moves with high accuracy but is weak and achieves low average score. Unregularized Hedge and RM ($\lambda = 0$) achieve high score but low human prediction accuracy. 
By contrast, piKL-HedgeBot with different $\lambda$ achieves a variety of highly favorable combinations of the two.
$\lambda=10^{-1}$ gives about the same top-1 accuracy as the IL policy but improves score by a factor of 1.4x over the IL anchor policy.
$\lambda=10^{-3}$ outperforms unregularized search methods in both score (by \textasciitilde 5\%) and human prediction accuracy (by \textasciitilde 6\%). Mild regularization improves average score, rather than harming it. 

\begin{figure*}
\begin{minipage}[b]{11.8cm}
\begin{figure}[H]
\def\sc{.56}%
\includegraphics[scale=\sc]{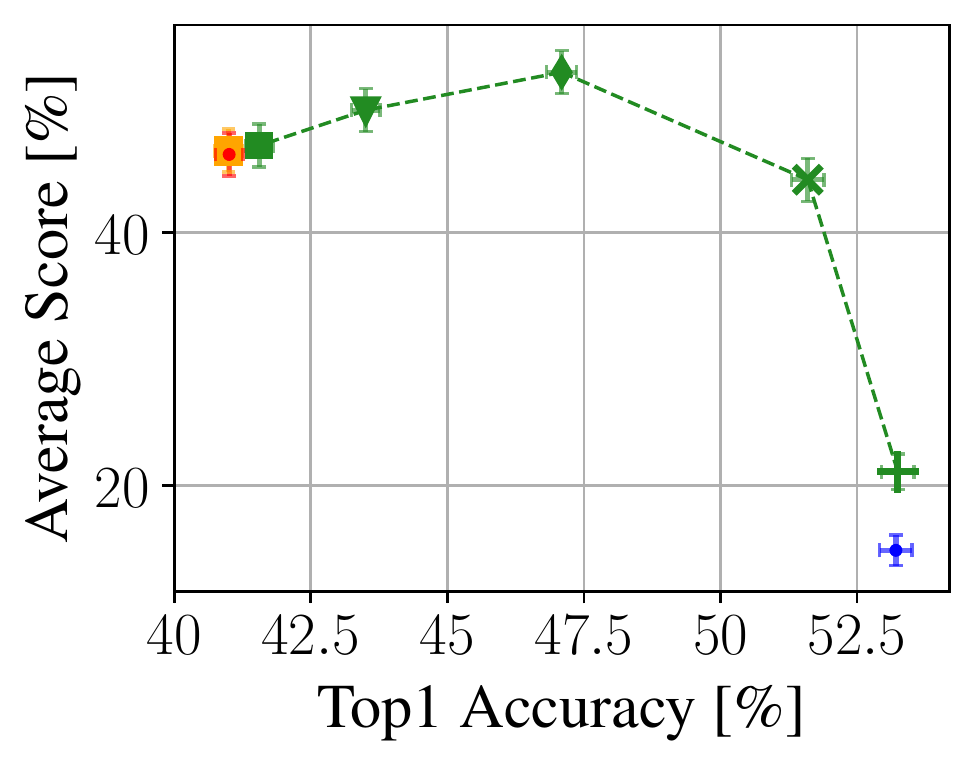}\hfill%
\includegraphics[scale=\sc]{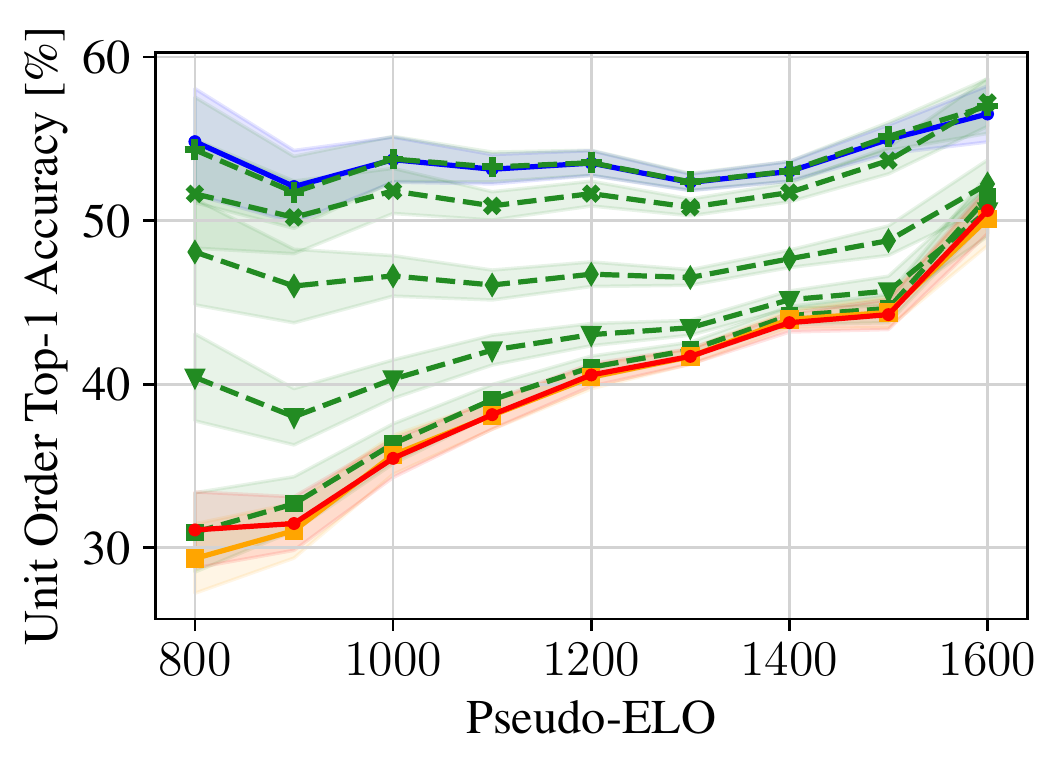}\\[-1mm]
\centering\includegraphics[scale=.57]{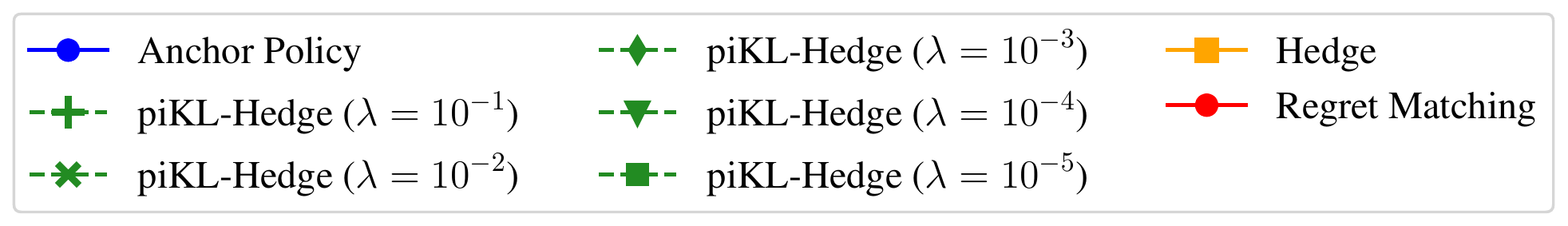}\\[-4mm]
\caption{\small (Left) Average top-1 accuracy of unit orders in each action predicted by the human IL anchor policy, RM, Hedge and piKL-Hedge, versus head to head score against 6 IL anchor policies. \qre~($\lambda = 10^{-1}$) predicts human moves as accurately as the anchor policy while achieving a much higher score. At the same time, \qre~($\lambda = 10^{-3}$) allows for a stronger and more human-like policy than unregularized search methods (hedge, RM). Note that equal performance would imply an average score of 1/7 $\approx 14.3\%$. Error bars indicate 1 standard error.
(Right) Average top-1 accuracy of unit orders in each action predicted by the human policy, RM, and piKL, as a function of pseudo-Elo player rating.
}
\label{fig:diplomacysosaccuracycurveL}
\label{fig:rating_1_nopress_allplaus_order_top1R}
\end{figure}
\end{minipage}\hfill%
\begin{minipage}[b]{5cm}
\begin{table}[H]\small\centering
\scalebox{1}{\begin{tabular}{lr}
 \textbf{Agent} & \makebox[1cm][r]{\textbf{Average Score}}\\
\toprule
DipNet$^\dagger$    & 3.7\% $\pm$ 0.3\% \\
DipNet RL$^\dagger$    & 4.7\% $\pm$ 0.3\% \\
\midrule
Blueprint$^\ddagger$     & 4.9\% $\pm$ 0.3\% \\
BRBot$^\ddagger$       & 16.1\% $\pm$ 0.6\% \\
SearchBot$^\ddagger$  & 13.4\% $\pm$ 0.5\% \\
\midrule
\bp   & 7.9\% $\pm$ 0.4\% \\
\sbot  & 31.3\% $\pm$ 0.7\% \\
\midrule
\qre      &      32.9\% $\pm$ 0.7\% \\
 ($\lambda=10^{-3}$)\\
\bottomrule
\end{tabular}}
\caption{\small Average score achieved by agents in a uniformly sampled pool of other agents. piKL-HedgeBot ($\lambda=10^{-3}$) outperforms all other agents in this setting. DipNet agents from \citep{paquette2019no} use a temperature of 0.1, while \bp and blueprint \citep{gray2020human} use a temperature of 0.5. The $\pm$ shows one standard error. $^\dagger$\citep{paquette2019no}; $^\ddagger$\citep{gray2020human}.}
\label{tab:population}
\end{table}
\end{minipage}
\vspace{-0.05in}
\end{figure*}
We also tested pure RL agents and found they perform poorly in predicting human moves. In particular, the recently proposed DORA and HumanDNVI-NPU algorithms~\cite{bakhtin2021no} achieve top-1 accuracy of only 29.1\% and 37.8\% respectively.

Next, in \cref{fig:rating_1_nopress_allplaus_order_top1R} (Right), we compare the top-1 accuracy of these methods across players of different pseudo-Elo ratings. The pseudo-Elo ($e_i$) for player i is constructed based on the logit rating $s_i$ introduced in \citet{gray2020human}, where, $e_i = \frac{s_i \cdot 400}{\log(10)} + 1000$. The top-1 accuracy for all the search-based policies increases with pseudo-Elo, indicating that they are better at modeling stronger players than weaker players. piKL-Hedge ($\lambda=10^{-1}$) performs just as well as the anchor policy across pseudo-Elos while being significantly stronger than the anchor, while $\lambda=10^{-3}$ is as strong or stronger than Hedge and RM but matches human play far better.

\xsubsubsection{piKL-HedgeBot performs well against a varied pool of agents}
\label{sec:results_piKL_pool}
We also develop a new head-to-head evaluation setting, where rather than testing one agent vs 6x of another agent, all 7 agents per game are sampled uniformly from a pool. %
The 1v6 head-to-head scores used in prior work \cite{gray2020human, bakhtin2021no, anthony2020learning, paquette2019no} indicate whether a population of 6x agents can be invaded by a 1x agent, and hence whether the 6 agents constitute an Evolutionarily Stable Strategy (ESS) \citep{taylor1978evolutionary, smith1982evolution}. By contrast, assigning the 7 agents per game randomly from a pool studies the robustness of an agent to a variety of other agents.%

We experiment with a pool of 8 agents. Five are previously published agents: DipNet, DipNet RL~\cite{paquette2019no}, Blueprint, BRBot, SearchBot~\cite{gray2020human} and three are our agents: IL policy, \sbot~and \qre. Doing well in this population requires playing well with both human-like policies (DipNet, Blueprint) and equilibrium policies (SearchBot, RMBot). Each experiment only compares one lambda value of piKL-HedgeBot for fairness.

The results of these experiment with piKL-HedgeBot ($\lambda=10^{-3}$) is presented in \cref{tab:population}. piKL-HedgeBot ($\lambda=10^{-3}$) outperforms all other agents. 

Overall, unlike Chess and Go, we do not find that piKL-Hedge clearly improves human prediction accuracy over the IL model. However, unlike Chess and Go, we observe that piKL-Hedge does improve playing strength over prior search methods against various past agents. In general-sum games like Diplomacy, it appears that piKL-hedge with a human anchor policy allows for slightly better play in a population containing human-like agents while still doing well against equilibrium-searchers, or alternatively can improve strength over IL to a lesser degree with no cost at all to prediction accuracy.

\vspace{-0.05in}
\section{Related Work}
\vspace{-0.05in}
\label{sec:relatedwork}

\subsection{Regularized Learning and Planning}
\vspace{-0.05in}

Several prior works have explored augmenting reinforcement learning with supervised data from expert demonstrations~\citep{vecerik2017leveraging,nair2018overcoming}. For example, \citet{hester2018deep} augment Deep Q Learning with a margin loss on demonstration data that aims to make $Q(a)$ for each demonstration action higher than that of other actions. KL-regularization has also been used successfully to incorporate expert demonstrations into RL training ~\citep{rudnerpathologies, ng2000algorithms, boularias2011relative, wu2019behavior, peng2019advantage, siegel2019keep}. In these settings, the standard RL objective is augmented by a KL divergence penalty that expresses the dissimilarity between the online policy and a reference policy derived from demonstrations. This helps guide exploration in RL or ameliorate inaccurate modeling of the environment in domains such as robotics. AlphaStar~\cite{vinyals2019grandmaster}, which achieved expert human performance in StarCraft~2, uses self-play RL initialized from supervised policies with a KL penalty term for deviating from the supervised policy during training to "aid in exploration and to preserve strategic diversity". In our work, we use the KL term to better approximate human play during inference-time search.

Prior work has explored entropy-regularized utilities in games. \citet{ling2018game} show that a particular type of entropic regularization in extensive-form games leads to quantal response equilibria. \citet{Cen21:Fast} give fast online optimization algorithms for entropy-regularized utilities in the context of quantal response. \citet{Farina19:Online} regularize utilities with a KL divergence term from a precomputed Nash equilibrium strategy to design agents that trade off game-theoretic safety and exploitation.  In our work, we leverage the KL divergence term towards a human-imitation learned policy instead, to regularize the utilities. And unlike previous works, we empirically study our approach in a much larger imperfect information game.

\vspace{-0.05in}
\subsection{Strong Human-Compatible Policies}
\vspace{-0.05in}
Prior work in multi-agent reinforcement learning has emphasized the importance of human-compatible policies in cooperative multi-agent environments. \citet{lerer2019learning} demonstrate that self-play policies may perform poorly with other agents if they do not conform to the population equilibrium (social conventions), and propose a combination of policy gradient and imitation loss directly on samples of population data. 

Human-compatible policies have also been studied on the benchmark game of Hanabi, where ad hoc play with humans is regarded as an open challenge problem~\cite{bard2020hanabi}. Most work on this challenge has focused on {\it zero-shot} coordination with humans, in which an agent must adapt to human play with no prior experience with human partners \cite{hu2020other,hu2021off,cui2021klevel}. Learning human-compatible policies from a combination of human data and planning are less well-studied in this setting.

\xsubsubsection{MCTS, Chess and Go}
\label{sec:inchessandgo}
Prior work in tree search methods, especially in chess and Go, has typically focused on developing strong agents without concern for accurately modeling human behavior.
For example in Go, a significant body of older work investigates imitation learning (IL) to obtain a baseline policy prior to use with MCTS, but tunes and evaluates the final agent via playing strength alone \cite{tian2016better, Cazenave2017}.

Regarding the use of search for human modeling, recent work in chess by \citet{mcilroy-young2020learning} found that pure IL outperformed all other approaches and that adding MCTS with the parameters of a standard engine significantly harmed human prediction accuracy. However in our work, using a different range of parameters, we show clear results to the contrary. In Go, \citet{Wang2017DANNLTE} report promising results on human prediction and playing strength using search, albeit with a specialized architecture and rollout method. \citet{Baier2018Emulating} report in the card game Spades both excellent human accuracy and playing strength by adding a human policy bias to a variant of MCTS. To our knowledge, our work is the first to demonstrate a clear gain in human prediction accuracy over deep learning models in the highly-studied domains of chess and Go via simple well-established methods of policy-regularized planning.

\xsubsubsection{Diplomacy}
\label{sec:indiplomacy}
Diplomacy is a benchmark 7-player game that involves communication, negotiation, cooperation, and competition in a strategic multi-agent setting. While chess and Go are two-player zero-sum games for which optimal play is well-defined and can be computed through self-play \cite{nash1951non}, Diplomacy has no such guarantees and strong play likely requires modeling other agents, even in the no-press variant where natural-language communication is not allowed~\cite{bakhtin2021no}.

\citet{paquette2019no} showed that neural network IL on human data in no-press Diplomacy can reasonably approximate human play, but that bootstrapping RL from this agent leads to a breakdown in cooperation. \citet{anthony2020learning} developed new RL methods based on fictitious play that improve the performance of agents in no-press Diplomacy, and \citet{gray2020human} showed that an equilibrium-finding regret minimization search procedure on top of human IL models achieves human-level performance in no-press Diplomacy. However, although both methods rely on the human IL model to generate a restricted action set for RL or search, neither contains any explicit regularization when choosing among those actions, and we show that the equilibrium search in \citet{gray2020human} greatly decreases the accuracy of modeling human players. Similarly, \citet{bakhtin2021no} achieved strong results in no-press Diplomacy via self-play RL both from scratch and initialized with a human-learned policy, but we show that the resulting final agents do not ultimately model humans well.

\vspace{-0.05in}
\section{Conclusion}
\vspace{-0.05in}

In this paper, we showed across several domains that regularizing search policies according to a KL-divergence loss with an imitation learned (IL) policy produces policies that maintain high human prediction accuracy while being far stronger than the original learned policy. In chess and Go, applying standard MCTS regularized toward a human-learned policy
achieves state-of-the-art prediction accuracy, surpassing imitation learning, while also winning more than 85\% of games against an IL model. 
We then introduced a novel regret minimization algorithm that is regularized based on the KL divergence from an IL policy. In no-press Diplomacy, this algorithm yields both a policy that predicts human play with the same accuracy as imitation learning alone while increasing win rate against state-of-the-art baselines by a factor of 1.4, or alternately a policy that outperforms unregularized search while achieving much higher human prediction accuracy. We presented in Appendix~\ref{appendix:hanabi} similar successful results for KL-regularized search in Hanabi.

There are several directions for future work, such as extending piKL-Hedge to handle extensive-form games. Additionally, there may be better ways to regularize search than KL-divergence. Finally, it remains to be seen how KL-regularized search performs when combined with RL.

\section*{Author Contributions}

A. P. Jacob was the primary researcher for piKL-hedge and contributed to the direction, experiments, and writing of the entire paper. D. J. Wu was the primary researcher for MCTS in chess and Go, and contributed to the direction, experiments, and writing of the entire paper. G. Farina was the primary formulator of the piKL-hedge algorithm and handled all the theory in the paper. A. Lerer contributed to the direction of the project, the formulation of piKL-hedge, its experimental evaluation, and paper writing. H. Hu was the primary researcher for the extension of piKL to Hanabi covered in Appendix~\ref{appendix:hanabi}. A. Bakhtin contributed to the experimental evaluation of piKL-hedge. J. Andreas contributed to the direction of the project and to paper writing. N. Brown initiated the project and contributed to the direction of the project, the formulation of piKL-hedge, its experimental evaluation, and paper writing.

\bibliography{references}
\bibliographystyle{icml2022}

\onecolumn
\appendix

\section{Proofs}
\label{sec:proofs}
\newcommand{\numberthis}[1]{\refstepcounter{equation}\tag{\theequation}\label{#1}}

\def\utilde{\tilde{u}}

In this Appendix, we present detailed proofs of \cref{prop:regret} and \cref{cor:exploitability}.

\subsection{Known results}

We start by recalling a few standard results. First, we recall the follow-the-regularized-leader (FTRL) algorithm, one of the most well-studied algorithms in online optimization. At every time $t$, the FTRL algorithm instantiated with domain $\mathcal{X}$, $1$-strongly-convex regularizer $\phi: \mathcal{X} \to \mathbb{R}$, and learning rate $\eta > 0$, produces iterates according to
\begin{equation}\tag{FTRL}\label{eq:ftrl}
    \vec{x}^{t+1} = \argmax_{\vec{x} \in \mathcal{X}} \mleft\{ -\frac{\phi(\vec{x})}{\eta} + \sum_{\tau=1}^t \ell^\tau(\vec{x}) \mright\},
\end{equation}
where $\vec{\ell}^1,\dots,\vec{\ell}^{t}$ are the convex utility functions gave as feedback by the environment. The FTRL algorithm guarantees the following regret bound.

\begin{lemma}[\citet{rakhlin09:lecture}, Corollary~7]\label{lem:ftrl}
    The iterates $\vec{x}^t \in \mathcal{X}$ produced by the FTRL algorithm set up with constant step size $\eta > 0$ and $1$-strongly convex regularizer $\phi$ satisfy the regret bound
    \[
        \sum_{t=1}^T \ell^t(\vec{u}) - \ell^t(\vec{x}^t) \le \frac{\phi(\vec{u})}{\eta} + \sum_{t=1}^T \ell^t(\vec{x}^{t+1}) - \ell^t(\vec{x}^t) \qquad\forall\,\vec{u} \in \mathcal{X}.
        \numberthis{eq:ftrl bound}
    \]
\end{lemma}

In the analysis of \cref{algo:noregret} we will also make use of the following technical lemma, a proof of which can be obtained starting using the same techniques as~\citet[Lemma~A.4]{abernethy09:beating}

\begin{lemma}\label{lem:abe}
    Let $\vec{p}\in\Delta(A)$ be a distribution over a discrete set $A$, $\vec{q}\in\mathbb{R}^{|A|}$ be a vector, and $D > 0$ be any constant such that $\max_{a,a'\in A} \{\vec{q}(a) - \vec{q}(a')\} \le M$. Then,
    \[
        \frac{\sum_{a \in A} \vec{p}(a)\cdot \exp\{-\vec{q}(a)\}^2}{
            \mleft(\sum_{a\in A} \vec{p}(a) \cdot \exp\{-\vec{q}(a)\}\mright)^2
        } - 1 \le \frac{\exp\{2M\}}{M^2} \sum_{a\in A} p(a) \vec{q}(a)^2.
    \]
\end{lemma}
\begin{proof}
    Let $q_\text{min} \defeq \min_{a\in A} \vec{q}(a)$ and $\tilde{\vec{q}}(a) \defeq \vec{q}(a) - q_\text{min}$ be a shifted version of $\vec{q}(a)$ so that $0 \le \tilde{\vec{q}}(a) \le M$ for all $a \in A$. Let now $X$ denote a random variable with value $\tilde{\vec{q}}(a)$ with probability $\vec{p}(a)$ for all $a \in A$. Then,
    \begin{align*}
        \frac{\sum_{a \in A} \vec{p}(a)\cdot \exp\{-\vec{q}(a)\}^2}{
            \mleft(\sum_{a\in A} \vec{p}(a) \cdot \exp\{-\vec{q}(a)\}\mright)^2
        } - 1 &= \frac{\exp\{q_\text{min}\}^2\sum_{a \in A} \vec{p}(a)\cdot\exp\{\vec{q}(a)\}^2}{
            \exp\{q_\text{min}\}^2\mleft(\sum_{a\in A} \vec{p}(a) \cdot \exp\{-\vec{q}(a)\}\mright)^2
        } - 1\\
        &= \frac{\sum_{a \in A} \vec{p}(a)\cdot\exp\{-\tilde{\vec{q}}(a)\}^2}{
            \mleft(\sum_{a\in A} \vec{p}(a) \cdot \exp\{-\tilde{\vec{q}}(a)\}\mright)^2
        } - 1\\
        &= \frac{\mathbb{E}[\exp(-X)]}{[\mathbb{E}\exp(-X)]^2} - 1. 
    \end{align*}
    Applying Lemma~A.4 from \citet{abernethy09:beating} we obtain
    \[
        \frac{\sum_{a \in A} \vec{p}(a)\cdot \exp\{-\vec{q}(a)\}^2}{
            \mleft(\sum_{a\in A} \vec{p}(a) \cdot \exp\{-\vec{q}(a)\}\mright)^2
        } - 1 \le \frac{\exp\{2M\} - 2M - 1}{M^2} (\mathbb{E}[X^2] - \mathbb{E}[X]^2) \le \frac{\exp\{2M\}}{M^2}\mathbb{E}[X^2],
    \]
    which is exactly the statement.
\end{proof}

\subsection{Bounding the distance between the iterates of \cref{algo:noregret}}

\begin{lemma}
    At all times $t$ and for all players $i$, the policies $\vec{\pi}_i^t$ produced by the FTRL algorithm set up with constant step size $\eta$ and negative entropy regularizer $\varphi(\vec{x}) \defeq \sum_{a\in A_i} \vec{x}(a)\log\vec{x}(a)$, when observing the utilities $\mathcal{U}^t_i$, match the policies $\vec{\pi}_i^t$ produced by \cref{algo:noregret}.
\end{lemma}
\begin{proof}
    Plugging the particular choices of utilities and regularizer into~\eqref{eq:ftrl}, we obtain
    \begin{align*}
        \vec{\pi}_{i}^{t+1} &= \argmax_{\vec{\pi}\in\Delta(A_i)}\mleft\{ \mleft(\sum_{t'=1}^{t}\mathcal{U}_{i}^{t'}(\vec{\pi})\mright) - \frac{1}{\eta}\sum_{a\in A_i} \vec{\pi}(a) \log \vec{\pi}(a)\mright\}\\
        &= \argmax_{\vec{\pi}\in\Delta(A_i)}\mleft\{\eta\mleft( \sum_{t'=1}^{t}\sum_{a \in A_i} \vec{\pi}(a)\,u_i(a, \vec{a}_{-i}^{t'}) - \lambda_i\,\vec{\pi}(a) \log\mleft(\frac{\vec{\pi}(a)}{\vec{\tau}_i(a)}\mright)\mright) - \sum_{a\in A_i} \vec{\pi}(a) \log \vec{\pi}(a)\mright\}\\
        &= \argmax_{\vec{\pi}\in\Delta(A_i)}\mleft\{ \eta\sum_{a\in A_i}\mleft(t\lambda_i \log\vec{\tau}_i(a) + \sum_{t'=1}^{t} u_i(a, \vec{a}_{-i}^{t'})\mright)\vec{\pi}(a) - (1+t\lambda_i\eta) \sum_{a\in A_i}\vec{\pi}(a)\log\vec{\pi}(a)\mright\}\\
        &= \argmax_{\vec{\pi}\in\Delta(A_i)}\mleft\{ \frac{\eta}{1+t\lambda_i\eta}\sum_{a\in A_i}\mleft(t\lambda_i \log\vec{\tau}_i(a) + \sum_{t'=1}^{t} u_i(a, \vec{a}_{-i}^{t'})\mright)\vec{\pi}(a) - \sum_{a\in A_i}\vec{\pi}(a)\log\vec{\pi}(a)\mright\}.\numberthis{eq:argmax}
    \end{align*}
    A well-known closed form solution to the above entropy-regularized problem is given by the softmax function. In particular, let
    \[
        \vec{w}_i^{t+1}(a) \defeq \frac{\eta}{1+t\lambda_i\eta} \mleft(t\lambda_i \log \vec{\tau}_i(a) + \sum_{t'=1}^t \utilde_i(a, \vec{a}_{-i}^{t'})\mright)\qquad\forall a\in A_i.\numberthis{eq:def w}
    \]
    Then,
    \[
        \vec{\pi}_i^{t+1}(a) = \frac{\exp\{\vec{w}_i^{t+1}(a)\}}{\sum_{a'\in A_i} \exp\{\vec{w}_i^{t+1}(a')\}} \qquad\forall a \in A_i,
    \]
    which coincides with the iterate produced by \cref{algo:noregret}.
\end{proof}

The next observation shows that the iterates $\vec{\pi}_i$ do not change if the utility function $u_i$ is first shifted to be in the range $[0,D_i]$.

\begin{remark}
    Consider the shifted utilities
    \[
        \utilde_i(a,\vec{a}_{-i}^t) \defeq u_i(a, \vec{a}_{-i}^{t'}) - \min_{\vec{a}\in A_1\times \dots\times A_n} u_i(\vec{a}) \in [0, D_i]\numberthis{eq:def utilde}
    \]
    and let $\vec{v}$ be defined as \eqref{eq:def w} using $\utilde_i$ in place of $u_i$, that is,
    \[
        \vec{v}_i^{t+1}(a) \defeq \frac{\eta}{1+t\lambda_i\eta} \mleft(t\lambda_i \log \vec{\tau}_i(a) + \sum_{t'=1}^t \utilde_i(a, \vec{a}_{-i}^{t'})\mright)\qquad\forall a\in A_i.\numberthis{eq:def v}
    \]
    Then, the iterates $\vec{\pi}_i$ can be equivalently expressed as
    \[
        \vec{\pi}_i^{t+1}(a) = \frac{\exp\{\vec{v}_i^{t+1}(a)\}}{\sum_{a'\in A_i} \exp\{\vec{v}_i^{t+1}(a')\}} \qquad\forall a \in A_i.
    \]
\end{remark}
\begin{proof}
    Let $\gamma \defeq \\min_{\vec{a}\in A_1\times \dots\times A_n} u_i(\vec{a})$ denote the minimum utility that Player~$i$ can get against the actions of the opponents. Since the argmax of a function does not change if a constant is added to the objective, from~\eqref{eq:argmax} we can write
    \begin{align*}
        \vec{\pi}_{i}^{t+1} &= \argmax_{\vec{\pi}\in\Delta(A_i)}\mleft\{ -\frac{\eta}{1+t\lambda_i\eta}\gamma + \frac{\eta}{1+t\lambda_i\eta}\sum_{a\in A_i}\mleft(t\lambda_i \log\vec{\tau}_i(a) + \sum_{t'=1}^{t} u_i(a, \vec{a}_{-i}^{t'})\mright)\vec{\pi}(a) - \sum_{a\in A_i}\vec{\pi}(a)\log\vec{\pi}(a)\mright\}\\
        &= \argmax_{\vec{\pi}\in\Delta(A_i)}\mleft\{ \frac{\eta}{1+t\lambda_i\eta}\sum_{a\in A_i}\mleft(t\lambda_i \log\vec{\tau}_i(a) + \sum_{t'=1}^{t} (u_i(a, \vec{a}_{-i}^{t'}) - \gamma)\mright)\vec{\pi}(a) - \sum_{a\in A_i}\vec{\pi}(a)\log\vec{\pi}(a)\mright\}\\
        &= \argmax_{\vec{\pi}\in\Delta(A_i)}\mleft\{ \sum_{a\in A_i}\vec{v}_i^{t+1}(a)\vec{\pi}(a) - \sum_{a\in A_i}\vec{\pi}(a)\log\vec{\pi}(a)\mright\},
    \end{align*}
    where the second equality follows from the fact that $\vec{\pi}\in\Delta(A_i)$. 
    
    A solution is again given by softmax function
    \[
        \vec{\pi}_i^{t+1}(a) = \frac{\exp\{\vec{v}_i^{t+1}(a)\}}{\sum_{a'\in A_i} \exp\{\vec{v}_i^{t+1}(a')\}} \qquad\forall a \in A_i.\numberthis{eq:pi fn v}
    \]
\end{proof}

In the rest of the proof we will use~\eqref{eq:pi fn v} to analyze the iterates $\vec{\pi}_i$ produced by the algorithm.

\begin{lemma}\label{lem:norm step}
    Let $\eta \le 1/(\lambda_i\beta_i + 2D_i)$. Then, at all times $t$,
    \[
        \|\vec{\pi}_i^{t+1} - \vec{\pi}_i^t\|_{\nabla^2 \varphi(\vec{\pi}_i^t)} \le \frac{\sqrt{3}\,e}{t\lambda_i\eta}.
    \]
\end{lemma}
\begin{proof}
    At all times $t$, introduce the vector $\vec{\xi}^t_i \in \mathbb{R}^{|A_i|}$ defined as
    \begin{equation}
        \vec{\xi}^{t}_i(a) \defeq \frac{\eta}{1+t\lambda_i\eta} \mleft(  - \lambda_i \vec{v}^t(a) + \lambda_i \log\vec{\tau}_i(a) + \utilde_i(a, \vec{a}_{-i}^t) \mright) \qquad\forall a\in A_i.\label{eq:def xi}
    \end{equation}
    At all times $t$ and for all $a$, it holds that
    \begin{align*}
        \vec{v}_i^{t+1}(a) &= \frac{\eta}{1+t\lambda_i\eta}\mleft(\frac{1+(t-1)\lambda_i\eta}{\eta}\vec{v}^t(a) + \lambda_i \log\vec{\tau}_i(a) + \utilde_i(a, \vec{a}_{-i}^t)\mright)\\
        &= \frac{\eta}{1+t\lambda_i\eta}\mleft(\frac{1+t\lambda_i\eta}{\eta}\vec{v}^t(a) - \lambda_i \vec{v}^t(a) + \lambda_i \log\vec{\tau}_i(a) + \utilde_i(a, \vec{a}_{-i}^t)\mright)\\
        &= \vec{v}^t(a) + \vec{\xi}_i^t(a).\numberthis{eq:v t plus 1}
    \end{align*}
    Substituting \eqref{eq:v t plus 1} we can write
    \[
        \vec{\pi}_i^{t+1}(a) = \frac{\exp\{\vec{v}_i^t(a)\}\cdot\exp\{\vec{\xi}_i^t(a)\}}{\sum_{a'\in A_i}\exp\{\vec{v}_i^t(a')\}\cdot\exp\{\vec{\xi}_i^t(a')\}} = \frac{\vec{\pi}_i^t(a) \exp\{\vec{\xi}_i^t(a)\}}{\sum_{a' \in A_i} \vec{\pi}_i^t(a') \exp\{\vec{\xi}_i^t(a')\}}.
        \numberthis{eq:pi t plus 1}
    \]
    Expanding the definition of the local norm induced by $\nabla^2 \varphi(\vec{\pi}_i^t)$ we find
    \begin{align*}
        \|\vec{\pi}_i^{t+1} - \vec{\pi}_i^t\|^2_{\nabla^2 \varphi(\vec{\pi}_i^t)} &= \sum_{a\in A_i} \frac{1}{\vec{\pi}_i^t(a)} \mleft(\vec{\pi}_i^{t+1}(a) - \vec{\pi}_i^t(a)\mright)^2\\
        &= \sum_{a\in A_i} \vec{\pi}_i^t(a) \mleft( \frac{\exp\{\vec{\xi}_i^t(a)\}}{\sum_{a' \in A_i} \vec{\pi}_i^t(a') \exp\{\vec{\xi}_i^t(a')\}} -1\mright)^2 \numberthis{eq:norm expansion 1}\\
        &= \sum_{a\in A_i} \vec{\pi}_i^t(a) \mleft[ \mleft(\frac{\exp\{\vec{\xi}_i^t(a)\}}{\sum_{a' \in A_i} \vec{\pi}_i^t(a') \exp\{\vec{\xi}_i^t(a')\}}\mright)^2 - 2\mleft(\frac{\exp\{\vec{\xi}_i^t(a)\}}{\sum_{a' \in A_i} \vec{\pi}_i^t(a') \exp\{\vec{\xi}_i^t(a')\}}\mright) +1\mright]\\
        &=\frac{\sum_{a\in A_i} \vec{\pi}_i^t(a)\exp\{\vec{\xi}_i^t(a)\}^2}{\mleft(\sum_{a' \in A_i} \vec{\pi}_i^t(a') \exp\{\vec{\xi}_i^t(a')\}\mright)^2} - 2\frac{\sum_{a\in A_i} \vec{\pi}_i^t(a)\exp\{\vec{\xi}_i^t(a)\}}{\sum_{a' \in A_i} \vec{\pi}_i^t(a') \exp\{\vec{\xi}_i^t(a')\}} +\sum_{a\in A_i} \vec{\pi}_i^t(a)\\
        &=\frac{\sum_{a\in A_i} \vec{\pi}_i^t(a)\exp\{\vec{\xi}_i^t(a)\}^2}{\mleft(\sum_{a' \in A_i} \vec{\pi}_i^t(a') \exp\{\vec{\xi}_i^t(a')\}\mright)^2} -1\numberthis{eq:norm expansion 2},
    \end{align*}
    where~\eqref{eq:norm expansion 1} follows from substituting~\eqref{eq:pi t plus 1}. We now apply \cref{lem:abe}, applied with $\vec{q} = \vec{\xi}_i^t$, $\vec{p} = \vec{\pi}_i^t$, and $A = A_i$. First, we study the range $D_i = \max_{a,a'\in A_i} \{\vec{\xi}_i^t(a) - \vec{\xi}_i^t(a')\}$ used in the statement of the Lemma. In particular, using~\eqref{eq:v t plus 1} we have
    \allowdisplaybreaks\begin{align*}
        \max_{a,a'\in A_i} \{\vec{\xi}_i^t(a) - \vec{\xi}_i^t(a')\} &= \max_{a,a'\in A_i} \{\vec{v}_i^{t+1}(a) - \vec{v}_i^{t}(a) - \vec{v}_i^{t+1}(a') + \vec{v}_i^{t}(a')\}\\
        &=\max_{a,a'\in A_i} \mleft\{(\log\vec{\tau}_i(a)-\log\vec{\tau}_i(a')) \cdot \mleft(\frac{t\lambda_i\eta}{1+t\lambda_i\eta} - \frac{(t-1)\lambda_i\eta}{1+(t-1)\lambda_i\eta}\mright) \mright.\\&\hspace{5cm}
            +\ \frac{\eta}{1+t\lambda_i\eta}\mleft(\sum_{t'=1}^t \utilde_i(a, \vec{a}_{-i}^{t'}) - \utilde_i(a',\vec{a}_{-i}^{t'})\mright) \\&\hspace{5cm}
            \mleft. -\ \frac{\eta}{1+(t-1)\lambda_i\eta}\mleft(\sum_{t'=1}^{t-1} \utilde_i(a, \vec{a}_{-i}^{t'}) - \utilde_i(a',\vec{a}_{-i}^{t'})\mright)\mright\}\\
        &=\max_{a,a'\in A_i} \mleft\{\frac{\lambda_i\eta}{(1+t\lambda_i\eta)(1+(t-1)\lambda_i\eta)}(\log\vec{\tau}_i(a)-\log\vec{\tau}_i(a')) \mright.\\&\hspace{4cm}
            +\ \frac{\lambda_i\eta^2}{(1+t\lambda_i\eta)(1+(t-1)\lambda_i\eta)}\mleft(\sum_{t'=1}^{t-1} \utilde_i(a, \vec{a}_{-i}^{t'}) - \utilde_i(a',\vec{a}_{-i}^{t'})\mright)\\&\hspace{4cm}
            \mleft.+\ \frac{\eta}{1+t\lambda_i\eta}(\utilde_i(a,\vec{a}_{-i}^t) - \utilde_i(a',\vec{a}_{-i}^t))\mright\}\\
        &\le\max_{a,a'\in A_i} \mleft\{\frac{\lambda_i\eta}{(1+t\lambda_i\eta)(1+(t-1)\lambda_i\eta)}(\log\vec{\tau}_i(a)-\log\vec{\tau}_i(a')) \mright\}\\&\hspace{2cm}
            +\ \max_{a,a'\in A_i} \mleft\{\frac{\lambda_i\eta^2}{(1+t\lambda_i\eta)(1+(t-1)\lambda_i\eta)}\mleft(\sum_{t'=1}^{t-1} \utilde_i(a, \vec{a}_{-i}^{t'}) - \utilde_i(a',\vec{a}_{-i}^{t'})\mright)\mright\}\\&\hspace{4cm}
            +\ \max_{a,a'\in A_i} \mleft\{\frac{\eta}{1+t\lambda_i\eta}(\utilde_i(a,\vec{a}_{-i}^t) - \utilde_i(a',\vec{a}_{-i}^t))\mright\}\\
        &\le \eta(\lambda_i\beta_i + 2D_i),\numberthis{eq:xi bound}
    \end{align*}
    where the first inequality follows from upper bounding the max of a sum with the sum of max of each term, and the second inequality follows from noting that
    \[
    \frac{\lambda_i\eta}{(1+t\lambda_i\eta)(1+(t-1)\lambda_i\eta)} \le \frac{\eta}{1+t\lambda_i\eta} \le \eta,
    \] and
    \[\frac{\lambda_i\eta^2}{(1+t\lambda_i\eta)(1+(t-1)\lambda_i\eta)} \le \frac{\eta}{t}
    \mleft(\sum_{t'=1}^{t-1} \utilde_i(a, \vec{a}_{-i}^{t'}) - \utilde_i(a',\vec{a}_{-i}^{t'})\mright) \le \frac{\lambda_i\eta^2}{t\lambda_i\eta}\cdot tD_i = \eta D_i.
    \]
    Applying \cref{lem:abe} to the right-hand side of~\eqref{eq:norm expansion 2} using the bound on the range of $\vec{\xi}_i^t$ shown in~\eqref{eq:xi bound} yields
    \begin{align*}
        \|\vec{\pi}_i^{t+1} - \vec{\pi}_i^t\|^2_{\nabla^2 \varphi(\vec{\pi}_i^t)} &\le \frac{\exp\{2 \eta(\lambda_i\beta_i + 2D_i)\}}{\eta^2(\lambda_i\beta_i + 2D_i)^2} \sum_{a\in A_i} \vec{\pi}_i^t(a)\mleft(\vec{\xi}_i^t(a)\mright)^2.\numberthis{eq:after lemma 2}
    \end{align*}
    Using the fact that any convex combination of values is upper bounded by the maximum value, we can further bound the right-hand side of~\eqref{eq:after lemma 2} as
    \begin{align*}
        \|\vec{\pi}_i^{t+1} - \vec{\pi}_i^t\|^2_{\nabla^2 \varphi(\vec{\pi}_i^t)}
        &\le \frac{\exp\{2 \eta(\lambda_i\beta_i + 2D_i)\}}{\eta^2(\lambda_i\beta_i + 2D_i)^2}\max_{a\in A_i} (\vec{\xi}_i^t(a))^2\\
        &= \frac{\exp\{2 \eta(\lambda_i\beta_i + 2D_i)\}}{\eta^2(\lambda_i\beta_i + 2D_i)^2}\max_{a\in A_i} \mleft(\vec{v}_i^{t+1}(a) - \vec{v}_i^t(a)\mright)^2,
    \end{align*}
    where the equality follows from~\eqref{eq:v t plus 1}. Hence, expanding the definition of $\vec{v}^t_i$ and $\vec{v}_i^{t+1}$,
    \begin{align*}
        \|\vec{\pi}_i^{t+1} - \vec{\pi}_i^t\|^2_{\nabla^2 \varphi(\vec{\pi}_i^t)}
        &\le\frac{\exp\{2 \eta(\lambda_i\beta_i + 2D_i)\}}{\eta^2(\lambda_i\beta_i + 2D_i)^2}\max_{a\in A_i} \mleft\{ \mleft(\frac{t\lambda_i\eta}{1+t\lambda_i\eta} - \frac{(t-1)\lambda_i\eta}{1+(t-1)\lambda_i\eta}\mright)\log\vec{\tau}_i(a) \mright.\\&\hspace{5cm}
                    \mleft.+\ \frac{\eta}{1+t\lambda_i\eta}\sum_{t'=1}^t \utilde_i(a, \vec{a}_{-i}^{t'}) - \frac{\eta}{1+(t-1)\lambda_i\eta}\sum_{t'=1}^{t-1} \utilde_i(a, \vec{a}_{-i}^{t'})\mright\}^2\\
        &= \frac{\exp\{2 \eta(\lambda_i\beta_i + 2D_i)\}}{\eta^2(\lambda_i\beta_i + 2D_i)^2}\max_{a\in A_i} \mleft\{\frac{\lambda_i\eta}{(1+t\lambda_i\eta)(1+(t-1)\lambda_i\eta)}\log\vec{\tau}_i(a) + \frac{\eta}{1+t\lambda_i\eta}\utilde_i(a,\vec{a}_{-i}^{t})\mright.\\&\hspace{6.9cm} - \mleft.\frac{\lambda_i\eta^2}{(1+t\lambda_i\eta)(1+(t-1)\lambda_i\eta)}\sum_{t'=1}^{t-1}\utilde_i(a,\vec{a}_{-i}^{t'})\mright\}^{\!2}\\
        &\le 3\frac{\exp\{2 \eta(\lambda_i\beta_i + 2D_i)\}}{\eta^2(\lambda_i\beta_i + 2D_i)^2}\max_{a\in A_i} \mleft\{\mleft(\frac{\lambda_i\eta}{(1+t\lambda_i\eta)(1+(t-1)\lambda_i\eta)}\log\vec{\tau}_i(a)\mright)^2 + \mleft(\frac{\eta}{1+t\lambda_i\eta}\utilde_i(a,\vec{a}_{-i}^{t})\mright)^2 \mright.\\&\hspace{6.7cm} + \mleft.\mleft(\frac{\lambda_i\eta^2}{(1+t\lambda_i\eta)(1+(t-1)\lambda_i\eta)}\sum_{t'=1}^{t-1}\utilde_i(a,\vec{a}_{-i}^{t'})\mright)^2\mright\}\\  
        &= \frac{3}{(1+t\lambda_i\eta)^2}\frac{\exp\{2 \eta(\lambda_i\beta_i + 2D_i)\}}{\eta^2(\lambda_i\beta_i + 2D_i)^2}\max_{a\in A_i} \mleft\{\mleft(\frac{\lambda_i\eta}{1+(t-1)\lambda_i\eta}\log\vec{\tau}_i(a)\mright)^2 + \mleft({\eta}\utilde_i(a,\vec{a}_{-i}^{t})\mright)^2 \mright.\\&\hspace{6.7cm} + \mleft.\mleft(\frac{\lambda_i\eta^2}{(1+(t-1)\lambda_i\eta)}\sum_{t'=1}^{t-1}\utilde_i(a,\vec{a}_{-i}^{t'})\mright)^2\mright\}\\
        &\le \frac{3}{(1+t\lambda_i\eta)^2}\frac{\exp\{2 \eta(\lambda_i\beta_i + 2D_i)\}}{\eta^2(\lambda_i\beta_i + 2D_i)^2}(\lambda^2\eta^2\beta_i^2 + 2\eta^2 D_i^2) \\
        &\le 3\frac{\exp\{2\eta(\lambda_i\beta_i + 2D_i)\}}{(1+t\lambda_i\eta)^2}\\
        &\le 3\frac{\exp\{2\eta(\lambda_i\beta_i + 2D_i)\}}{(t\lambda_i\eta)^2} = \mleft(\sqrt{3}\cdot\frac{\exp\{\eta(\lambda_i\beta_i + 2D_i)\}}{t\lambda_i\eta}\mright)^2.
    \end{align*}
    Using the hypothesis that $\eta \le 1/(\lambda_i \beta_i + 2D_i)$ and taking square roots yields the statement.
\end{proof}

\subsection{Completing the analysis}

\begin{restatable}{proposition}{propregret}\label{prop:regret}
    Fix a player $i \in \{1, \dots, n\}$. The regret
    \[
        R_{i}^T \defeq \max_{\vec{\pi}^* \in \Delta(A_i)}\sum_{t=1}^T\mathcal{U}_{i}^t(\vec{\pi}^*) - \sum_{t=1}^T\mathcal{U}_{i}^t(\vec{\pi}^t_{i})
    \]
    incurred up to any time $T$ by policies $\vec{\pi}^t_{i}$ defined in~\eqref{eq:distribution} where the learning rate is set to any value $0 < \eta \le 1/(\lambda_i\beta_i + 2D_i)$, satisfies
    \[
        R_{i}^T \le \frac{\log |A_i|}{\eta} + \frac{3\,e ( 1+ \log T)}{\lambda_i\eta}(D_i+\lambda_i\beta_i+\lambda_i\sqrt{|A_i|}),
    \]
    where $D_i$ is any upper bound on the range of the possible rewards of Player~$i$, and
    \begin{equation}
        \beta_i \defeq \max_{a\in A_i} \log(1/\vec{\tau}(a)).
        \label{eq:def beta i}
    \end{equation}
\end{restatable}
\begin{proof}
    Let
    \[
        \vec{q}^t_i \defeq \Big(\utilde_i(a, \vec{a}^t_{-i})\Big)_{a\in A_i},
    \]
    and note that, by definition of the regularized utilities $\mathcal{U}_i$,
    \begin{align*}
        \mathcal{U}_i^t(\vec{\pi}^{t+1}_i) - \mathcal{U}_i^t(\vec{\pi}^t_i)
        &=  \vec{q}_i^\top \mleft(\vec{\pi}^{t+1}_i - \vec{\pi}^t_i\mright) - \lambda_i \KL{\vec{\pi}^{t+1}_i}{\vec{\tau}} + \lambda_i\KL{\vec{\pi}^t_i}{\vec{\tau}}\\
        &=  \vec{q}_i^\top \mleft(\vec{\pi}^{t+1}_i - \vec{\pi}^t_i\mright) - \lambda_i \varphi(\vec{\pi}^{t+1}_i) + \lambda_i\varphi(\vec{\pi}^t_i) - \lambda_i \nabla\varphi(\vec{\tau}_i)^\top (\vec{\pi}^t_i - \vec{\pi}^{t+1}_i)\\
        &=  \mleft(\vec{q}_i - \nabla \varphi(\vec{\tau}_i)\mright)^\top \mleft(\vec{\pi}^{t+1}_i - \vec{\pi}^t_i\mright) - \lambda_i \varphi(\vec{\pi}^{t+1}_i) + \lambda_i\varphi(\vec{\pi}^t_i)\\
        &\le  \mleft(\vec{q}_i + \nabla \varphi(\vec{\tau}_i)\mright)^\top \mleft(\vec{\pi}^{t+1}_i - \vec{\pi}^t_i\mright) - \lambda_i \varphi(\vec{\pi}^{t}_i) - \lambda_i \nabla\varphi(\vec{\pi}^t_i)^\top (\vec{\pi}^{t+1}_i - \vec{\pi}^t_i) + \lambda_i\varphi(\vec{\pi}^t_i)\\
        &=  \mleft(\vec{q}_i + \lambda_i\nabla \varphi(\vec{\tau}_i) - \lambda_i\nabla\varphi(\vec{\pi}^t_i)\mright)^\top \mleft(\vec{\pi}^{t+1}_i - \vec{\pi}^t_i\mright)\\
        &\le \mleft\|\vec{q}_i + \lambda_i\nabla \varphi(\vec{\tau}_i) - \lambda_i\nabla\varphi(\vec{\pi}^t_i)\mright\|_{\nabla^{-2}\varphi(\vec{\pi}^t_i)}\cdot \mleft\|\vec{\pi}^{t+1}_i - \vec{\pi}^t_i\mright\|_{\nabla^2\varphi(\vec{\pi}^t_i)},
    \end{align*}    
    where the first inequality follows by convexity and the second inequality by the generalized Cauchy-Schwarz inequality with the primal-dual norm pair $\|\cdot\|_{\nabla^2\varphi(\vec{\pi}^t_i)}$ and $\|\cdot\|_{\nabla^{-2}\varphi(\vec{\pi}^t_i)}$.
    A bound for the second term in the product if given by \cref{lem:norm step}. We now bound the first norm. First,
    \begin{align*}
        \mleft\|\vec{q}_i + \lambda_i\nabla \varphi(\vec{\tau}_i) - \lambda_i\nabla\varphi(\vec{\pi}^t_i)\mright\|_{\nabla^{-2}\varphi(\vec{\pi}^t_i)}^2 &= \sum_{a\in A_i} \vec{\pi}^t(a) \cdot\mleft( \tilde u_i(a, \vec{a}_{-i}^t) + \lambda_i \log \vec{\tau}_i(a) - \lambda_i \log\vec{\pi}^t_i(a)\mright)^2        \\
        &\le 3\sum_{a\in A_i} \vec{\pi}^t(a) \cdot\mleft( \tilde u_i(a, \vec{a}_{-i}^t)^2 + \lambda_i^2 (\log \vec{\tau}_i(a))^2 + \lambda_i^2 (\log\vec{\pi}^t_i(a))^2\mright)\\
        &\le 3\mleft(D_i^2 + \lambda_i^2\beta_i^2 + \lambda_i^2|A_i|\mright)\\
        &\le 3\mleft(D_i + \lambda_i\beta_i + \lambda_i \sqrt{|A_i|}\mright)^2,
    \end{align*}
    where the second inequality follows from the fact that $x \log^2 x \le 1$ for all $x\in [0,1]$.
    Taking square roots, we find
    \[
        \mleft\|\vec{q}_i + \lambda_i\nabla \varphi(\vec{\tau}_i) - \lambda_i\nabla\varphi(\vec{\pi}^t_i)\mright\|_{\nabla^{-2}\varphi(\vec{\pi}^t_i)} \le \sqrt{3}\mleft(D_i + \lambda_i\beta_i + \lambda_i \sqrt{|A_i|}\mright).
    \]
    So, using~\cref{lem:norm step}, we can write
    \[
    \mathcal{U}_i^t(\vec{\pi}^{t+1}_i) - \mathcal{U}_i^t(\vec{\pi}^t_i) \le \frac{3\,e}{t\lambda_i\eta}(D_i+\lambda_i\beta_i+\lambda_i\sqrt{|A_i|}).
    \]
    Plugging in the above expression into \cref{lem:ftrl} yields
    \begin{align*}
        R^T_i &\le \frac{\log |A_i|}{\eta} + \sum_{t=1}^T \frac{3\,e}{t\lambda_i\eta}(D_i+\lambda_i\beta_i+\lambda_i\sqrt{|A_i|})\\
        &\le \frac{\log |A_i|}{\eta} + \frac{3\,e ( 1+ \log T)}{\lambda_i\eta}(D_i+\lambda_i\beta_i+\lambda_i\sqrt{|A_i|}),
    \end{align*}
    which is the statement.
\end{proof}

\subsection{Proof of \cref{cor:distance}}
\thmdistancefromtau*
\begin{proof}
    By definition of regret, 
    \begin{align*}
        \frac{1}{T}R_i^T &= \frac{1}{T}\max_{{\vec{\pi}}^*_i\in\Delta(A_i)}\mleft\{\sum_{t=1}^T \mathcal{U}^t_i({\vec{\pi}}^*_i) - \mathcal{U}^t_i({\vec{\pi}}^t_i) \mright\}\\
        &=\max_{{\vec{\pi}}^*_i\in\Delta(A_i)}\mleft\{\mleft(\frac{1}{T}\sum_{t=1}^T u_i({\vec{\pi}}^*_i, \vec{\pi}_{-i}^t)\mright) - \mleft(\frac{1}{T}\sum_{t=1}^T u_i({\vec{\pi}}^t_i, \vec{\pi}_{-i}^t)\mright) - \frac{\lambda_i}{T}\sum_{t=1}^T \KL{\vec{\pi}_i^*}{\vec{\tau}_i} + \frac{\lambda_i}{T}\sum_{t=1}^T \KL{\vec{\pi}_i^t}{\vec{\tau}_i} \mright\}\\
        &=\max_{{\vec{\pi}}^*_i\in\Delta(A_i)}\mleft\{\mleft(\frac{1}{T}\sum_{t=1}^T u_i({\vec{\pi}}^*_i, \vec{\pi}_{-i}^t)\mright) - \mleft(\frac{1}{T}\sum_{t=1}^T u_i({\vec{\pi}}^t_i, \vec{\pi}_{-i}^t)\mright) - \lambda_i \KL{\vec{\pi}_i^*}{\vec{\tau}_i} + \frac{\lambda_i}{T}\sum_{t=1}^T \KL{\vec{\pi}_i^t}{\vec{\tau}_i} \mright\}\\
        &\ge\max_{{\vec{\pi}}^*_i\in\Delta(A_i)}\mleft\{\mleft(\frac{1}{T}\sum_{t=1}^T u_i({\vec{\pi}}^*_i, \vec{\pi}_{-i}^t)\mright) - \mleft(\frac{1}{T}\sum_{t=1}^T u_i({\vec{\pi}}^t_i, \vec{\pi}_{-i}^t)\mright) + \frac{\lambda_i}{T}\sum_{t=1}^T \KL{\vec{\pi}_i^t}{\vec{\tau}_i} \mright\}\\
        &\ge\max_{{\vec{\pi}}^*_i\in\Delta(A_i)}\mleft\{\mleft(\frac{1}{T}\sum_{t=1}^T u_i({\vec{\pi}}^*_i, \vec{\pi}_{-i}^t)\mright) - \mleft(\frac{1}{T}\sum_{t=1}^T u_i({\vec{\pi}}^t_i, \vec{\pi}_{-i}^t)\mright) + \lambda_i  \KL{\bar{\vec{\pi}}_i^T}{\vec{\tau}_i} \mright\}\\
        &=\max_{{\vec{\pi}}^*_i\in\Delta(A_i)}\mleft\{\mleft(\frac{1}{T}\sum_{t=1}^T (u_i({\vec{\pi}}^*_i, \vec{\pi}_{-i}^t)- u_i({\vec{\pi}}^t_i, \vec{\pi}_{-i}^t)\mright) + \lambda_i  \KL{\bar{\vec{\pi}}_i^T}{\vec{\tau}_i} \mright\}\\
        &\ge \max_{{\vec{\pi}}^*_i\in\Delta(A_i)}\mleft\{\mleft(\frac{1}{T}\sum_{t=1}^T -D_i\mright) + \lambda_i  \KL{\bar{\vec{\pi}}_i^T}{\vec{\tau}_i} \mright\} = -D_i + \lambda_i  \KL{\bar{\vec{\pi}}_i^T}{\vec{\tau}_i},
    \end{align*}
    where the first inequality holds since the KL divergence is nonnegative, the second inequality by convexity of the KL divergence, and the third inequality by definition of $D_i$. Rearranging yields the inequality in the statement.
\end{proof}
    
\subsection{Relationship with Nash Equilibrium}

In this subsection, we show that when all players play according to \cref{algo:noregret} in a \emph{two-player zero-sum} game, then the average policies $\bar{\vec{\pi}}_{i}^T$ converge to a Nash equilibrium \emph{of the regularized game whose utilities are $\mathcal{U}_i$}. 

\begin{restatable}{proposition}{propne}\label{prop:ne}
    For any $T \in \mathbb{N}$, $\eta > 0$, and $\delta \in (0,1)$, define the quantity
    \[
        \xi^T(\delta) \defeq
        \frac{R_1^T+ R_2^T}{T} + (\max_i D_i)\sqrt{\frac{32}{T}\log\frac{2\max_i |A_i|}{\delta}}.
    \]
    Upon running \cref{algo:noregret} for any $T$ iterations with learning rate $\eta > 0$, the average policies
    $\bar{\vec{\pi}}_{i}^T$ of each player form a $\xi^T(\delta)$-approximate Nash equilibrium with respect to the regularized utility functions $\mathcal{U}_{i}$ with probability at least $1-\delta$, for any $\delta \in (0,1)$.
\end{restatable}
\begin{proof}
    \newcommand{\ut}[2]{\mathcal{U}_{i}(#1, #2)}
    \newcommand{\utp}[2]{\mathcal{U}_{-i}(#2, #1)}
    \newcommand{\tildepi}[1]{{\vec{\pi}}^t_{#1}}
    \newcommand{\pit}{\vec{\pi}^t_{i}}
    Fix a player $i\in\{1,2\}$, and any policy $\vec{\pi}^*\in\Delta(A_i)$, and introduce the discrete-time stochastic process 
    \begin{align*}
        w^t &\defeq \Big(\ut{\vec{\pi}^*}{\tildepi{-i}} - \ut{\pit}{\tildepi{-i}}\Big) - \Big(\ut{\vec{\pi}^*}{a_{-i}^t} - \ut{\pit}{a_{-i}^t}\Big).
    \end{align*}
    Since the opponent player $-i$ plays according to \cref{algo:noregret}, its action $a_{-i}^t$ at all times $t$ is selected by sampling (unbiasedly) an action from the policy $\vec{\pi}^t_{-i}$. Therefore, $w^t$ is a martingale difference sequence. Furthermore, by expanding the definition of $\mathcal{U}_i$, the absolute value of $w^t$ satisfies
    \begin{align*}
        |w^t| &= \mleft|\Big(u_i(\vec{\pi}^*,  \tildepi{-i}) - u_i(\pit, \tildepi{-i})\Big) - \Big(u_i(\vec{\pi}^*, a_{-i}^t) - u_i(\pit, a_{-i}^t)\Big)\mright| \\
        &\le \mleft|u_i(\vec{\pi}^*,  \tildepi{-i}) - u_i(\pit, \tildepi{-i})\mright| - \mleft|u_i(\vec{\pi}^*, a_{-i}^t) - u_i(\pit, a_{-i}^t)\mright| \le 2D_i.
    \end{align*}
    Hence, using Azuma-Hoeffding's inequality, for any $\delta \in (0,1)$,
    \begin{align*}
        1-\delta &\le \mathbb{P}\mleft[ \sum_{t=1}^T w^t \le D_i\sqrt{8T \log \frac{1}{\delta}}\mright]\\
        &= \mathbb{P}\mleft[ \mleft(\sum_{t=1}^T \ut{\vec{\pi}^*}{\tildepi{-i}} - \sum_{t=1}^T \ut{\pit}{\tildepi{-i}}\mright) - \mleft(\sum_{t=1}^T u_i(\vec{\pi}^*, a^t_{-i}) - \sum_{t=1}^T \ut{\pit}{a_{-i}^t}\mright) \le \sqrt{8T \log \frac{1}{\delta}}\mright]\\
        &= \mathbb{P}\mleft[ \sum_{t=1}^T \ut{\vec{\pi}^*}{\tildepi{-i}} - \sum_{t=1}^T \ut{\pit}{\tildepi{-i}} \le R^T_{i} + D_i\sqrt{8T \log \frac{1}{\delta}}\mright],
    \end{align*}
    where $R^T_{i}$ is as defined in \cref{prop:regret}.
    Since the above expression holds for any $\vec{\pi}^*\in\Delta(A_i)$, in particular, using the union bound on each $a \in A_i$,
    \begin{equation}\label{eq:regret tilde}
        \mathbb{P}\mleft[ \max_{\vec{\pi}^* \in \Delta(A_i)}\sum_{t=1}^T \ut{\vec{\pi}^*}{\tildepi{-i}} - \sum_{t=1}^T \mathcal{U}_i(\vec{\pi}_{i}^t, {\vec{\pi}}_{-i}^t) \le R^T_{i} + D_i\sqrt{8T \log \frac{|A_i|}{\delta}}\mright] \ge 1-\delta
    \end{equation}
    for any player $i\in\{1,2\}$ and any $\delta\in (0,1)$.
    
    Summing Inequality~\eqref{eq:regret tilde} for $i\in\{1,2\}$ and using the union bound, we can further write
    \begin{align}
        &\mathbb{P}\mleft[ \max_{\vec{\pi}_1^* \in \Delta(A_1)}\mleft\{\sum_{t=1}^T \mathcal{U}_1(\vec{\pi}_1^*,\tildepi{2})\mright\} +
        \max_{\vec{\pi}_2^* \in \Delta(A_2)}\mleft\{\sum_{t=1}^T \mathcal{U}_2(\tildepi{1},\vec{\pi}_2^*)\mright\}
        - \mleft(\sum_{t=1}^T \mathcal{U}_1(\vec{\pi}_{1}^t, {\vec{\pi}}_{2}^t) + \mathcal{U}_2(\vec{\pi}_{1}^t, {\vec{\pi}}_{2}^t)\mright)\mright. \nonumber\\
        &\hspace{7cm}\mleft.\le R^T_{1} + R^T_2 + (\max_i D_i)\sqrt{32T \log \frac{\max_i |A_i|}{\delta}}\mright] \ge 1-2\delta.\nonumber
    \end{align}
    Dividing by $T$ and noting that for any player $i\in\{1,2\}$
    \begin{align*}
        \frac{1}{T} \sum_{t=1}^T \ut{\vec{\pi}^*}{\tildepi{-i}} &= \mathcal{U}_{i}\mleft(\vec{\pi}^*, \frac{1}{T}\sum_{t=1}^T \tildepi{-i}\mright) = \mathcal{U}_{i}\mleft(\vec{\pi}^*, \bar{\vec{\pi}}^T_{-i}\mright)
    \end{align*}
    further yields
    \begin{align}
        &\mathbb{P}\mleft[ \max_{\vec{\pi}_1^* \in \Delta(A_1)}\mleft\{\mathcal{U}_1(\vec{\pi}_1^*,\bar{\vec{\pi}}^T_{2})\mright\} +
        \max_{\vec{\pi}_2^* \in \Delta(A_2)}\mleft\{\mathcal{U}_2(\bar{\vec{\pi}}^T_1,\vec{\pi}_2^*)\mright\}
        - \frac{1}{T}\mleft(\sum_{t=1}^T \mathcal{U}_1(\vec{\pi}_{1}^t, {\vec{\pi}}_{2}^t) + \mathcal{U}_2(\vec{\pi}_{1}^t, {\vec{\pi}}_{2}^t)\mright)\mright. \nonumber\\
        &\hspace{10cm}\mleft.\le \frac{R^T_{1} + R^T_2}{T} + D_i\sqrt{\frac{32}{T} \log \frac{\max_i |A_i|}{\delta}}\mright] \ge 1-2\delta.\label{eq:regret sum 2}
    \end{align}
    We now analyze the term in parenthesis, that is, 
    \[
        (\clubsuit) \defeq -\frac{1}{T}\mleft(\sum_{t=1}^T \mathcal{U}_1(\vec{\pi}_{1}^t, {\vec{\pi}}_{2}^t) + \mathcal{U}_2(\vec{\pi}_{1}^t, {\vec{\pi}}_{2}^t)\mright)
    \]
    Plugging in the definition of $\mathcal{U}_1$ and $\mathcal{U}_2$, that is,
    \begin{align*}
        \mathcal{U}_1(\vec{\pi}_1,\vec{\pi}_2) &\defeq u_1(\vec{\pi}_1,\vec{\pi}_2) - \lambda_1 \KL{\vec{\pi}_1}{\vec{\tau}_1}\\
        \mathcal{U}_2(\vec{\pi}_1,\vec{\pi}_2) &\defeq u_2(\vec{\pi}_1,\vec{\pi}_2) - \lambda_2 \KL{\vec{\pi}_2}{\vec{\tau}_2} = -u_1(\vec{\pi}_1,\vec{\pi}_2) + \lambda_2 \KL{\vec{\pi}_2}{\vec{\tau}_2}
    \end{align*}
    into~$(\clubsuit)$ yields
    \begin{align}
        (\clubsuit) = -\frac{1}{T}\mleft(\sum_{t=1}^T \mathcal{U}_1(\vec{\pi}_{1}^t, {\vec{\pi}}_{2}^t) + \mathcal{U}_2(\vec{\pi}_{1}^t, {\vec{\pi}}_{2}^t)\mright) &= \frac{1}{T}\mleft(\sum_{t=1}^T\lambda_1 \KL{\tildepi{1}}{\vec{\tau}_1} + \lambda_2\KL{\tildepi{2}}{\vec{\tau}_2}\mright) \nonumber\\
        &\ge \lambda_1 \KL{\bar{\vec{\pi}}_1^T}{\vec{\tau}_1} + \lambda_2 \KL{\bar{\vec{\pi}}_2^T}{\vec{\tau}_2}\label{eq:step1x}\\
        &= -\mleft(\mathcal{U}_1(\bar{\vec{\pi}}_1^T, \bar{\vec{\pi}}_2^T) + \mathcal{U}_2(\bar{\vec{\pi}}_1^T, \bar{\vec{\pi}}_2^T)\mright),\label{eq:step2x}
    \end{align}
    where~\eqref{eq:step1x} follows from convexity of the KL divergence function, and~\eqref{eq:step2x} follows again from the definition of $\mathcal{U}_1$ and $\mathcal{U}_2$. Substituting~\eqref{eq:step2x} back into~\eqref{eq:regret sum 2}, we find
    \begin{align}
        &\mathbb{P}\mleft[ \max_{\vec{\pi}_1^* \in \Delta(A_1)}\mleft\{\mathcal{U}_1(\vec{\pi}_1^*,\bar{\vec{\pi}}^T_{2}) - \mathcal{U}_1(\bar{\vec{\pi}}^T_1, \bar{\vec{\pi}}^T_2)\mright\} +
        \max_{\vec{\pi}_2^* \in \Delta(A_2)}\mleft\{\mathcal{U}_2(\bar{\vec{\pi}}^T_1,\vec{\pi}_2^*) - \mathcal{U}_2(\bar{\vec{\pi}}^T_1, \bar{\vec{\pi}}^T_2)\mright\}\mright. \nonumber\\
        &\hspace{7cm}\mleft.
        \le \frac{R^T_{1} + R^T_2}{T} + (\max_i D_i)\sqrt{\frac{32}{T} \log \frac{\max_i |A_i|}{\delta}}\mright] \ge 1-2\delta.
    \end{align}
    Since
    \[
    \max_{\vec{\pi}_1^* \in \Delta(A_1)}\mleft\{\mathcal{U}_1(\vec{\pi}_1^*,\bar{\vec{\pi}}^T_{2}) - \mathcal{U}_1(\bar{\vec{\pi}}^T_1, \bar{\vec{\pi}}^T_2)\mright\} \ge 0,\quad\text{and}\quad
    \qquad
    \max_{\vec{\pi}_2^* \in \Delta(A_2)}\mleft\{\mathcal{U}_2(\bar{\vec{\pi}}^T_1,\vec{\pi}_2^*) - \mathcal{U}_2(\bar{\vec{\pi}}^T_1, \bar{\vec{\pi}}^T_2)\mright\} \ge 0,
    \]
    the inequality above in particular implies that 
    \begin{align}
        &\mathbb{P}\mleft[ \max\mleft\{\max_{\vec{\pi}_1^* \in \Delta(A_1)}\mleft\{\mathcal{U}_1(\vec{\pi}_1^*,\bar{\vec{\pi}}^T_{2}) - \mathcal{U}_1(\bar{\vec{\pi}}^T_1, \bar{\vec{\pi}}^T_2)\mright\},
        \max_{\vec{\pi}_2^* \in \Delta(A_2)}\mleft\{\mathcal{U}_2(\bar{\vec{\pi}}^T_1,\vec{\pi}_2^*) - \mathcal{U}_2(\bar{\vec{\pi}}^T_1, \bar{\vec{\pi}}^T_2)\mright\}\mright\}\mright. \nonumber\\
        &\hspace{7cm}\mleft.
        \le \frac{R^T_{1} + R^T_2}{T} + (\max_i D_i)\sqrt{\frac{32}{T} \log \frac{\max_i |A_i|}{\delta}}\mright] \ge 1-2\delta,
    \end{align}
    which is equivalent to the statement after making the variable substitution $\delta \defeq \delta'/2$.
\end{proof}

In particular, when $\eta \le 1/(\lambda_i\beta_i + 2D_i)$ for both players $i\in\{1,2\}$, \cref{prop:ne} implies that the average strategy profile is a $O(1/\sqrt{T})$-Nash equilibrium with respect to the regularized utility functions $\mathcal{U}_i$. 

A standard application of the Borel-Cantelli lemma enables to convert from the high-proability guarantees of \cref{prop:ne} at finite time to almost-sure convergence in the limit. Specifically,
\begin{corollary}\label{cor:limit is nash}
 Let $(\bar{\vec{\pi}}_1,\bar{\vec{\pi}}_2)$ be any limit point of the average policies $(\bar{\vec{\pi}}_{1}^T, \bar{\vec{\pi}}_2^T)$ of the players. Almost surely, $(\bar{\vec{\pi}}_1,\bar{\vec{\pi}}_2)$ is a Nash equilibrium with respect to the regularized utility functions $\mathcal{U}_{1},\mathcal{U}_2$, respectively.
\end{corollary}

From there, it is immediate to give guarantees with respect to the original (\emph{i.e.}, unregularized) game, and \cref{cor:exploitability} follows.
\thmexploitability*
\begin{proof}
    From \cref{cor:limit is nash}, almost surely $(\bar{\vec{\pi}}_1,\bar{\vec{\pi}}_2)$ is a Nash equilibrium of the regularized game whose players' utilities are $\mathcal{U}_1$ and $\mathcal{U}_2$, respectively. Expanding the definition of Nash equilibrium relative to Player~$1$, we have that
    \begin{align*}
        0 &=\max_{\vec{\pi}^*_1 \in \Delta(A_1)}\mleft\{\mathcal{U}_1(\vec{\pi}^*_1, \bar{\vec{\pi}}_2) - \mathcal{U}_1(\bar{\vec{\pi}}_1, \bar{\vec{\pi}}_2)\mright\}\\
        &= \max_{\vec{\pi}^*_1 \in \Delta(A_1)}\mleft\{u_1(\vec{\pi}^*_1, \bar{\vec{\pi}}_2) - \lambda_1\KL{\vec{\pi}^*_1}{\vec{\tau}_1} - u_1(\bar{\vec{\pi}}_1, \bar{\vec{\pi}}_2) + \lambda_1\KL{\bar{\vec{\pi}}_1}{\vec{\tau}_1} \mright\}\\
        &\ge \max_{\vec{\pi}^*_1 \in \Delta(A_1)}\mleft\{u_1(\vec{\pi}^*_1, \bar{\vec{\pi}}_2) - u_1(\bar{\vec{\pi}}_1, \bar{\vec{\pi}}_2) \mright\} - \lambda_1\KL{\vec{\pi}^*_1}{\vec{\tau}_1}\\
        &= \max_{\vec{\pi}^*_1 \in \Delta(A_1)}\mleft\{(u_1(\vec{\pi}^*_1, \bar{\vec{\pi}}_2) - u_1(\bar{\vec{\pi}}_1, \bar{\vec{\pi}}_2))  - \lambda_1\sum_{a\in A_1} \vec{\pi}^*_i(a) \log(\vec{\pi}^*_i(a))-\lambda_1\sum_{a\in A_1} \vec{\pi}^*_1(a) \log(1/ \vec{\tau}_1(a))\mright\}\\
        &\ge \max_{\vec{\pi}^*_1 \in \Delta(A_1)}\mleft\{(u_1(\vec{\pi}^*_1, \bar{\vec{\pi}}_2) - u_1(\bar{\vec{\pi}}_1, \bar{\vec{\pi}}_2)) -\lambda_1\sum_{a\in A_1} \vec{\pi}^*_1(a) \log(1/ \vec{\tau}_1(a))\mright\}\\
        &\ge \max_{\vec{\pi}^*_1 \in \Delta(A_1)}\mleft\{u_1(\vec{\pi}^*_1, \bar{\vec{\pi}}_2) - u_1(\bar{\vec{\pi}}_1, \bar{\vec{\pi}}_2)\mright\} - \lambda_1\beta_1,
    \end{align*}
    where the first inequality follows since the KL divergence is alwasy nonnegative, the second inequality since the negative entropy function is nonpositive on the simplex, and the third inequality follows from the definition of $\beta_1$. Symmetrically, for Player~$2$ we find that
    \[
        0 \ge \max_{\vec{\pi}^*_2 \in \Delta(A_2)}\mleft\{u_2(\bar{\vec{\pi}}_1, \vec{\pi}^*_2) - u_2(\bar{\vec{\pi}}_1, \bar{\vec{\pi}}_2)\mright\} - \lambda_2\beta_2.
    \]
    Hence, the exploitability of $\bar{\vec{\pi}}_1$ is at most $\lambda_1\beta_1$, while the exploitability of $\bar{\vec{\pi}}_2$ is at most $\lambda_2\beta_2$, which immediately implies the statement.
\end{proof}

\section{Illustrations of piKL-Hedge in Blotto}
\label{appendix:blotto}
\begin{figure}[ht!]
\centering
\includegraphics[width=6.8cm]{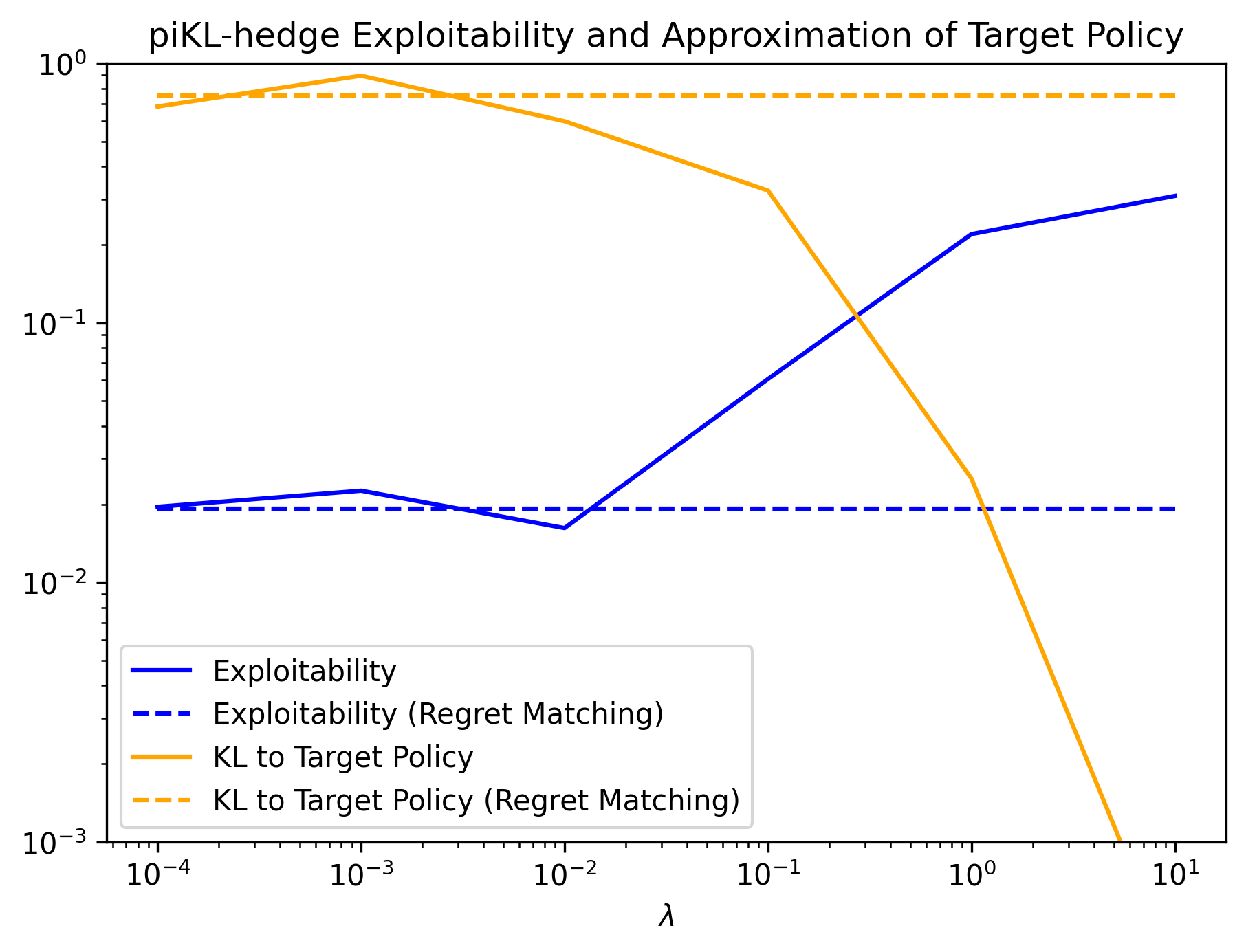}
\vspace{-0.2in}
\caption{\small Comparison of piKL and regret matching as a function of $\lambda$ in Colonel Blotto(10, 3). As $\lambda$ increases in piKL-hedge, piKL moves closer to the anchor policy at the cost of increased exploitability. The scale of $\lambda$ is related to the scale of the payoffs in the game, which are [0, 1] in Blotto. }
\label{fig:blotto}
\end{figure}

Colonel Blotto is a famous 2-player simultaneous action game that has a large action space but has rules that are short and simple. In Blotto, each player has $c$ coins to be distributed across $f$ fields. The aim is to win the most fields by allocating the player's coins across the fields. A field is won by contributing the most coins to that field (and drawn if there is a tie). The winner receives a reward of +1 and the loser receives -1. Both receive 0 in the case of a tie.

In Figure \ref{fig:blotto} we illustrate the key features of piKL-Hedge in Blotto, using a uniform anchor policy for convenience. Incidentally, piKL-Hedge with a uniform anchor policy converges to a quantal response equilibrium \cite{mckelvey1995quantal}. piKL-Hedge finds policies that play close to the anchor policy while having low regret, with $\lambda$ controlling the relative optimality of these two desiderata.

\section{Human Policy KL-Regularized Search also improves Cross-Entropy}
\label{appendix:piklcrossentropy}

\begin{table*}[ht]
\setlength{\tabcolsep}{0.5\tabcolsep}
\small
  \begin{center}
    \begin{tabular}{cc|QQQQQQ}
\toprule
Model & Dataset & (raw model) & $c_\text{puct} = 10$ & $c_\text{puct} = 5$ & $c_\text{puct} = 2$ & $c_\text{puct} = 1$ & $c_\text{puct} = 0.5$ \\
\midrule
\setcounter{MinX}{146}
\setcounter{MidX}{147}
\setcounter{MaxX}{150}
Maia1500 (Chess) & Lichess 1500 Rating Bucket & 1.476 & 1.469 & 1.465 & 1.470 & 1.504 & 1.598  \\
\setcounter{MinX}{142}
\setcounter{MidX}{144}
\setcounter{MaxX}{147}
Maia1900 (Chess) & Lichess 1900 Rating Bucket & 1.440 & 1.429 & 1.422 & 1.418 & 1.443 & 1.529  \\
\setcounter{MinX}{136}
\setcounter{MidX}{139}
\setcounter{MaxX}{141}
Our Model (Go) & GoGoD & 1.388 & 1.362 & 1.359 & 1.362 & 1.391 & 1.478 \\ 
\bottomrule
    \end{tabular}
    \caption{Cross-entropy predicting human moves in chess and Go using smooth KL optimization post-processing of MCTS with various $c_\text{puct}$. Largest standard error of any value is around 0.0033.}
   \label{table:go-cross-entropy}
  \end{center}
\end{table*}

\begin{figure*}[ht!]
\centering
\includegraphics[width=0.9\textwidth]{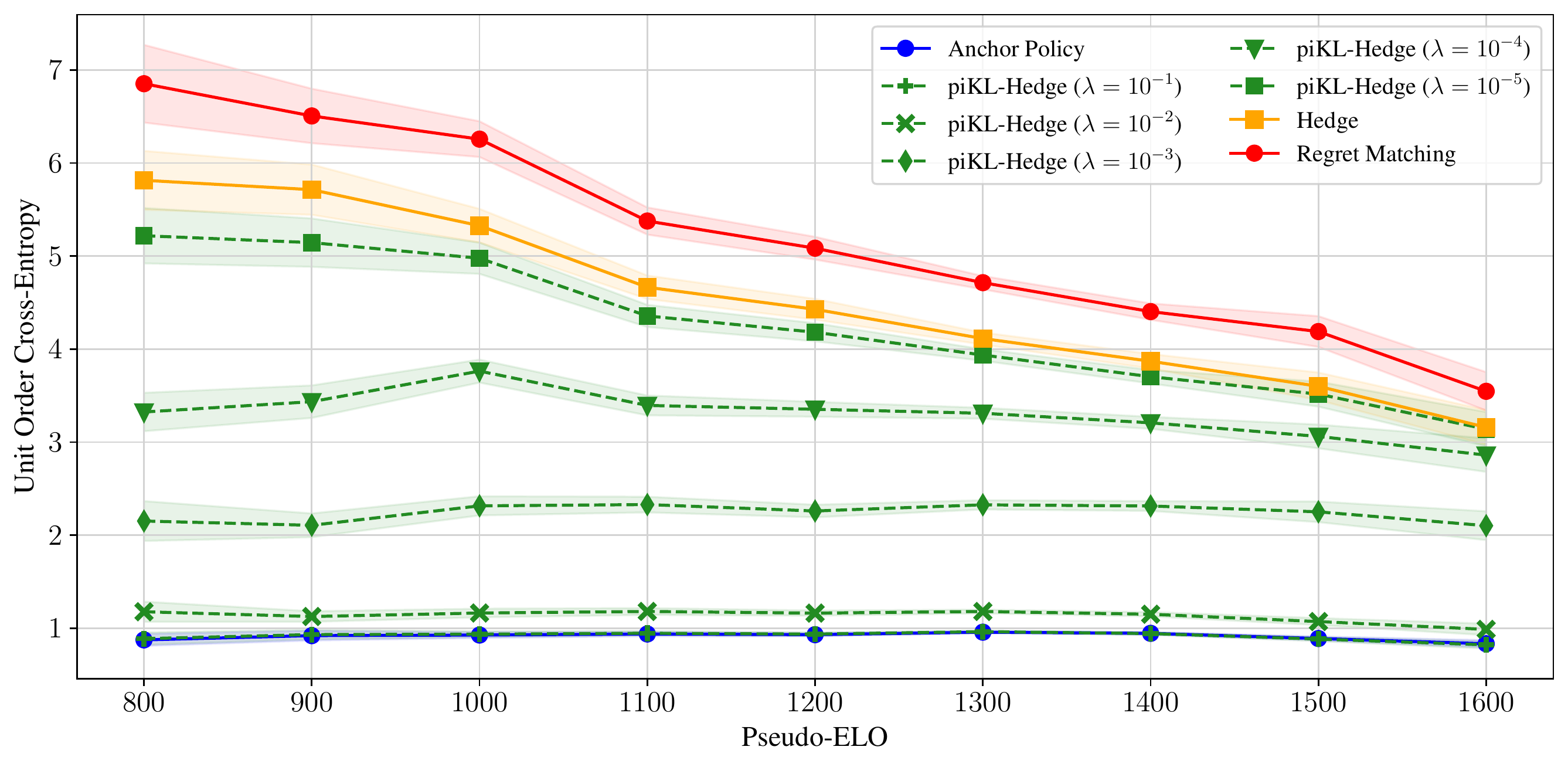}
\vspace{-0.2in}
\caption{\small Average cross-entropy of per-unit order prediction in Diplomacy, comparing the human-imitation-learned anchor policy, regret matching, hedge, and piKL-Hedge, as a function of pseudo-Elo player rating.}
\label{fig:rating_1_nopress_allplaus_avg_order_xe}
\end{figure*}

Here we show that not only does policy-regularized search improve top-1 accuracy, it also improves cross entropy for a reasonable range of parameter values in both chess and Go, as well as performing well in Diplomacy.

For Chess and Go, to obtain the policy distribution with which to compute its cross entropy with the human data, unlike for top-1 accuracy or for play we cannot directly use the MCTS visit distribution because occasionally simply due to discretization MCTS may give zero visits to the actual move that a human played, resulting in an undefined (i.e. infinite) cross-entropy. Instead, we leverage the result of \citet{grill2020monte} that PUCT-style MCTS with a policy prior can be seen as a discrete approximation to solving a smooth optimization:
$$ \argmax_{\pi} \sum_a Q(s,a)\pi(s,a) + \lambda \KL{\tau}{\pi} $$
where
$$\lambda = c_\text{puct} \frac{\sqrt{\sum_a n_a}}{(k + \sum_a n_a)}$$

where $Q$ is the current value estimate from search for each action, $\tau$ is the anchor or prior policy,  $\lambda$ controls the strength of the regularization towards that prior as a function of the number of visits, $c_\text{puct}$ is the MCTS exploration coefficient, $n_a$ is the number of times $a$ was explored and $k$ is an arbitrary constant not affecting the asymptotic results. 

We perform MCTS exactly the same as normal, the only difference is that at the very end, rather than using visit counts, we compute $\pi$ optimizing the above objective using the human imitation-learned anchor policy for $\tau$ and the MCTS-estimated Q-values for $Q$ (and using the same unweighted average child value for moves that lack a Q-value estimate due to having zero visits as described in Appendix \ref{appendix:mcts}). We use this resulting smooth $\pi$ as the final policy prediction and compute its cross entropy with the actual human moves.

Note that the authors choose $k = |A|$, the number of legal actions, for mathematical convenience, corresponding to adding one extra visit per every possible action~\citep{grill2020monte}, but since in our experiments we use only use 50 visits including the root visit (i.e. $\sum_a n_a = 49$), and the branching factor can be as large as 362 in Go, adding one extra visit per legal action in our case greatly overestimates the total number of visits, which in practice gives a less accurate correspondence between MCTS and this smoother regularized solution, so we instead choose $k = 0$.

In Table \ref{table:go-cross-entropy} we show the results. Across roughly the same parameter ranges, the regularized search policy using this smoothed MCTS postprocessing achieves lower cross-entropy with human moves than the raw imitation-learned policy without search. This suggests that not only does search improve on the raw imitation-learned policy at pinpointing the top action, it also gives a more accurate model of the overall distribution of likely human actions. 

Similarly, in \cref{fig:rating_1_nopress_allplaus_avg_order_xe}, we show that piKL-Hedge achieves better average cross entropy of unit orders in no-press Diplomacy compared to unregularized search methods, and matches that of imitation learning at $\lambda = 0.1$. This corroborates the results of Section \ref{sec:result_strong_human} and shows that piKL-Hedge provides the same benefits in modeling the overall distribution of human actions as it does on predicting the top move - outpredicting unregularized search (while playing as well or slightly better against human-like opponents), or equaling the prediction quality of imitation learning (while playing much more strongly than IL).

\section{More Experiments in Chess and Go}
\label{appendix:chessgo_experiments}

\begin{table}[H]
\centering\small
\begin{tabular}{lrrr}
\toprule
\multirow{2}{*}{\bf Game} & \multirow{2}{*}{$c_\text{puct}$} & \multicolumn{2}{c}{\bf MCTS Win\% vs raw model} \\
 &  & $\text{temp}=1$ & $\text{temp}=0.3$ \\ 
\midrule
Chess & 10.0 & 72.2\% $\pm$ 1.2\%  & 62.4\% $\pm$ 1.3\%  \\
Chess & 5.0 & 79.5\% $\pm$ 1.1\%  & 72.3\% $\pm$ 1.2\%  \\ 
Chess & 2.0 & 88.0\% $\pm$ 0.8\% & 86.1\% $\pm$ 0.9\% \\ 
Chess & 1.0 & 92.2\% $\pm$ 0.7\%  & 92.9\% $\pm$ 0.6\%  \\ 
Chess & 0.5 & 94.4\% $\pm$ 0.6\%  & 94.7\% $\pm$ 0.5\% \\ 
\midrule
Go & 10.0 & 73.2\% $\pm$ 1.4\%  & 63.3\% $\pm$ 1.5\%  \\
Go & 5.0 & 80.5\% $\pm$ 1.3\% & 74.1\% $\pm$ 1.4\%  \\
Go & 2.0 & 87.6\% $\pm$ 1.0\%  & 85.3\% $\pm$ 1.1\% \\
Go & 1.0 & 94.6\% $\pm$ 0.7\%  & 94.4\% $\pm$ 0.7\%  \\
Go & 0.5 & 96.4\% $\pm$ 0.6\%  & 97.0\% $\pm$ 0.5\%  \\
\bottomrule
    \end{tabular}
    \caption{Winrate of base model + MCTS vs base model at temperature 1 and 0.3. Base model is Maia1900 in chess, and our GoGoD model in Go. 1000 games per figure, draws count as half a win, $\pm$ indicates one standard error. Go uses Japanese rules with 6.5 komi. MCTS greatly improves strength in Chess and Go, the smallest $c_\text{puct}$ values improve it most. %
    }
    \label{table:maia-playing-strength}
\end{table}

\begin{table*}[h]
\setlength{\tabcolsep}{0.5\tabcolsep}
\small
  \begin{center}
    \begin{tabular}{lrrrrrrr}
\toprule
Model & Predicting & Main Time & Increment & Raw Model Acc \% & MCTS Acc \% & \textbf{Acc Gain from MCTS} & Approx Stderr \\
\midrule
\setcounter{MinX}{0}%
\setcounter{MidX}{15}%
\setcounter{MaxX}{90}%
Maia1500 & 1500	& 3m	& 0s	& 51.9	& 52.1	& \textbf{0.2} & 0.17 \\
Maia1500 & 1500	& 5m	& 0s	& 52.7	& 53.2	& \textbf{0.5} & 0.16 \\
Maia1500 & 1500	& 10m	& 0s	& 52.4	& 53.1	& \textbf{0.7} & 0.20 \\
Maia1500 & 1500	& 3m	& 2s	& 53.1	& 53.7	& \textbf{0.6} & 0.31 \\
Maia1500 & 1500	& 5m	& 3s	& 52.5	& 53.4	& \textbf{0.9} & 0.35 \\
Maia1500 & 1500	& 15m	& 15s	& 51.9	& 52.6	& \textbf{0.7} & 0.37 \\
\midrule						
\setcounter{MinX}{0}%
\setcounter{MidX}{75}%
\setcounter{MaxX}{200}%
Maia1900 & 1900	& 3m	& 0s	& 53.0	& 53.9	& \textbf{0.8} & 0.12 \\
Maia1900 & 1900	& 5m	& 0s	& 53.2	& 54.5	& \textbf{1.3} & 0.17 \\
Maia1900 & 1900	& 10m	& 0s	& 52.9	& 54.6	& \textbf{1.7} & 0.24 \\
Maia1900 & 1900	& 3m	& 2s	& 54.1	& 55.5	& \textbf{1.4} & 0.27 \\
Maia1900 & 1900	& 5m	& 3s	& 53.1	& 55.1	& \textbf{2.0} & 0.51 \\
Maia1900 & 1900	& 15m	& 15s	& 53.7	& 55.9	& \textbf{2.2} & 0.56 \\

\bottomrule
    \end{tabular}
    \caption{Difference in top-1 accuracy between raw model and MCTS using the best $c_\text{puct}$ in predicting human moves for chess players in rating buckets 1500 and 1900 using Maia1500 and Maia1900, split by time control of the games, excluding all time controls with fewer than 100 games. Approx Stderr indicates the rough standard error of raw accuracy values on that row given the number of games of that time control. Despite the statistical uncertainty on some individual values, overall the improvement of MCTS vs the raw model does clearly tend to be larger on the games with longer time controls.}
   \label{table:maia-by-tc}
  \end{center}
\end{table*}

This section summarizes the results of a small number of additional experiments in chess and Go. Whereas Figure \ref{fig:gochesscurve} used a temperature of 1.0 when sampling from the agent policy, we show in Table \ref{table:maia-playing-strength} that MCTS also achieves similar winrates versus the raw model when both are sampled using a much lower temperature of 0.3.

Additionally, we confirm in Go the same result that \citet{mcilroy2020aligning} reported in chess, that current top RL agents are far worse at matching human moves than even just imitation learning. We tested the current top open-source Go program KataGo~\citep{Wu2020Go} using its final 6, 10, and 15 block models\footnote{from https://katagotraining.org} which range from upper-amateur to superhuman level, and found that they achieve accuracies of about 35\%, 43\%, and 43\%, all more than 10\% lower than the results in Table \ref{table:chess-and-go-accuracy}.

Lastly, in Table \ref{table:maia-by-tc} we show that in chess the amount of improvement given by MCTS over the raw model in predicting human players on the Lichess test set tends to be larger for games with longer time controls than for shorter time controls, going from about 0.2\% to about 0.7\% between the shortest and longer time controls for 1500-1599-rated players, and going from about 0.8\% to around 2.0\% for 1900-1999-rated players. This is consistent with the intuition that humans rely more heavily on planning when they have more time to think, increasing the gain from modeling that planning.

\clearpage
\section{Baseline Model Architecture and Training for Go}
\label{appendix:gotraining}

As summarized in Section \ref{sec:chessgoarchitecture}, for training baseline imitation-learning models to play on the 19x19 board in Go, our architecture follows the same 20-block 256-channel residual net described in \citet{silver2017mastering}, except with the addition of squeeze-and-excitation layers at the end of each residual block~\citep{hu2018SE}. In particular, the following additional operations are inserted just prior to each skip connection that adds the output $R$ of a residual block to the trunk $X$:
\begin{itemize}
    \item Channelwise global average pooling of $R$ from 19 $\times$ 19 $\times$ 256 channels to 256 channels.
    \item A fully connected layer including bias from 256 channels to 64 channels.
    \item A ReLU nonlinearity.
    \item A fully connected layer including bias from 64 channels to 512 channels, which are split two vectors $S$ and $B$ of 256 channels each.
    \item Output $R$ $\times$ sigmoid($S$) + $B$ to be added back to the trunk $X$, instead of $R$ as in a normal residual net.
\end{itemize}

In other words, the final result of the residual block as a whole is ReLU($X$ + $R$ $\times$ sigmoid($S)$ + $B$) instead of ReLU($X$ + $R$).

Additionally, some games are played with a \emph{komi} (compensation given to White for playing second) that is not equal to the value of 7.5 used by \citet{silver2017mastering}. Therefore, for the final feature of the input encoding, rather than a binary-valued feature equalling 1 if the player to move is White and 0 if the player to move is Black, we instead use the real-valued feature of komi/10 if the player to move is White or -komi/10 if the player to move is Black. We additionally exclude a very tiny number of games with extreme komi values, outside of the range [-60,60].

We train using a mini-batch size of 2048 distributed as 8 batches of 256 across 8 GPUs, and train for a total of 64 epochs - roughly 475000 minibatches for the GoGoD dataset. We use SGD with momentum 0.9, weight decay coefficient of 1e-4, and a learning rate schedule of of 1e-1, 1e-2, 1e-3, 1e-4 for the first 16, next 16, next 16, and last 16 epochs respectively. 

We train both the policy and values heads jointly, minimizing the cross entropy of the policy head with respect to the one-hot move made in the actual game, and the MSE of the value head with respect to the game result of -1 or 1, except similarly to \citet{silver2017mastering} we weight the MSE value loss by 0.01 to avoid overfitting of the value head.

\section{MCTS Algorithmic Details}
\label{appendix:mcts}
In this appendix, we summarize the details of the version of MCTS used in our experiments, including one often-overlooked detail. We follow a standard MCTS implementation very similar to that of \cite{silver2017mastering}.

Each turn, the algorithm builds and expands a game tree over multiple iterations rooted at the current state for that turn. On each iteration $t$ MCTS starts at the root and descends the tree by exploring at each state $s$ an action $a$ according to some exploration method. Upon reaching a state $s_t$ not yet explored, it adds $s_t$ to the tree, queries the value function $V_i(s_t)$ for each $i$ to estimate the total expected future reward, and updates the statistics of all nodes traversed based on $V_i(s_t)$ and any intermediate rewards received. Subsequent iterations begin again from the root. For our work in chess and Go, we follow the convention where win, loss, and draw have reward 1,-1, and 0.

The statistics tracked at the node for each state $s$ where player $i$ is to move include the visit counts $N(s,a)$ which are the number of iterations that reached state $s$ and tried action $a$, and $Q(s,a)$ the average value of those iterations from the perspective of player $i$, i.e. $Q(s,a) = (1/N(s,a)) \sum_{t} V_i(s_t) + U_i(s,s_t)$ where the sum ranges only over those iterations $t$ that reached state $s$ and tried action $a$ and $U_i(s,s_t)$ is the total intermediate reward on the path from $s$ to $s_t$. When descending the tree, the exploration method is to always select the action:
\begin{equation}
\argmax_a \, Q(s,a) + c_\text{puct}\tau(s,a)\frac{\sqrt{\sum_b N(s,b)}}{N(s,a)+1}
\end{equation}
where $\tau(s,a)$ is the prior policy probability for action $a$ in state $s$, and $c_\text{puct}$ is a tunable parameter controlling the tradeoff between exploration and exploitation. The final agent policy $\pi$ is simply proportional to the visit counts for the root, i.e. $\pi(s,a) = N(s,a)/\sum_b N(s,b)$ where $s$ is the root state, or optionally we may also have $\pi(s,a) \sim N(s,a)^{1/T}$ where $T$ is a temperature parameter.

One final often-overlooked detail concerns how to evaluate states for which there is no $Q(s,a)$ estimate. Since the tree policy depends on the $Q(s,a)$ estimates of the possible actions, there is a nontrivial choice of what $Q$ value to use for an action $a$ that has been tried \emph{zero} times and therefore never estimated. Since the tree branches exponentially, deeper in the MCTS tree there will always be many actions with zero visits, and so this choice can affect the behavior of MCTS even in the limit of large amounts of search. Unfortunately, the details of this choice have sometimes been left undiscussed and undocumented in major past work, and as a result major MCTS implementations have not standardized on it, variously choosing game-loss~\citep{AZFPU}, the current running average parent $Q$ or a $Q$-value minus a heuristic offset parameter~\citep{ELFLZFPU}, or many other options. 

In our work, we use the equal-weighted average value of all actions at the parent node that have been visited at least once, i.e. $\sum_a Q(s,a) I(N(s,a) > 0) / \sum_a I(N(s,a) > 0)$ . This can be viewed as corresponding to a naive prior that the values of actions are i.i.d draws from an unknown distribution. While not perfect, our choice is simple, parameter-free, behaves in a way that is invariant to any global translation or scaling of the game's rewards, and works reasonably in practice for our purposes.

\section{Brief Description of Diplomacy}
\label{appendix:rules_diplomacy}
We briefly summarizing the rules of Diplomacy. See~\citet{paquette2019no} for a more detailed description. The board is a map of Europe partitioned into 75 regions, 34 of which are \emph{supply centers} (SCs) that players compete to control. Players command multiple units and each turn privately issue orders for each unit they own (to hold, move, support another unit, or convoy). These orders are revealed at the same time, thereby making Diplomacy a simultaneous-action game. A player wins the game by controlling a majority (18) of the SCs.
A game may also end in a draw if all remaining players agree. In this case, we use the \textbf{Sum-of-Squares (SoS)} scoring system as used in prior works \cite{paquette2019no,gray2020human,bakhtin2021no}. If no player wins, SoS defines the score of player~$i$ as $C_i^2 /\sum_{i'}C_{i'}^2$, where $C_i$ is the SC count for player~$i$.

Diplomacy is specifically designed so that a player is unlikely to achieve victory without help from other players. The full game allows unrestricted private natural-language communication between players
each turn prior to choosing orders,
but we focus on the simpler \emph{no-press} variant, in which no such communication is allowed, but modeling how opponents will behave
continues to be important.

\section{Diplomacy hyper-parameters}
\label{appendix:hyperparams}
We describe the describe the search parameters used in this work and compare it to those in previous works. As compared to \citet{gray2020human}, we use a much less expensive set of search parameters for all results in the main section of this paper (See, \cref{tab:search_params_cmp}). In \cref{tab:old_compare_agents}, we show that these tuned-down set of parameters, slightly reduces piKL-HedgeBot's performance but allows us to make similar conclusions when comparing against SearchBot \cite{gray2020human} and supervised learning-based bots from prior works. Additionally, in \cref{tab:old_compare_agents}, we also show that when our agent (piKL-HedgeBot ($\lambda = 10^{-3}$)) uses the same search parameters as \citet{gray2020human}, it outperforms SearchBot by a big margin.

In our experiments, to compare against DipNet~\citep{paquette2019no} we use the original model checkpoint\footnote{DipNet SL from \url{https://github.com/diplomacy/research}. 
} and we sample from the policy with temperature $0.1$ (as used in prior works). Similarly, to compare against SearchBot~\citep{gray2020human} agent we use the released checkpoint\footnote{blueprint from \url{https://github.com/facebookresearch/diplomacy_searchbot/releases/tag/1.0}. 
} and agent configuration\footnote{\url{https://github.com/facebookresearch/diplomacy_searchbot/blob/master/conf/common/agents/searchbot_02_fastbot.prototxt}}.

The only other search parameter unique and new to our \qre~algorithm is $\eta$ in Algorithm \ref{algo:noregret}, which we heuristically set on each Hedge iteration $t$ to $c / (\sigma \sqrt{t})$ where $\sigma$ is the standard deviation across iterations of the average utility experienced by the agent $i$ being updated. We find $c=10/3$ works well for no-press Diplomacy.

\begin{table}[ht]
\small
\center
\begin{tabular}{lr}
\toprule
\bf Model & \bf Temperature \\
\midrule
DipNet \citep{paquette2019no} & 0.1 \\
DipNet RL \citep{paquette2019no} & 0.1 \\
Blueprint \citep{gray2020human} & 0.1 \\
IL Policy (Ours) & 0.5 \\
\bottomrule
\end{tabular}
\caption{Sampling temperatures used in the models across prior works. For our imitation learning model (IL Policy), we use a temperature of 0.5 in all experiments to encourage stochasticity.}
\label{tab:sl_temperature}
\end{table}

\begin{table*}[ht]
\small
\center
\begin{tabular}{lrr}
\toprule
\bf Parameter &\bf  \cite{gray2020human} & \bf Ours \\
\midrule
Number candidate actions ($N_c$) & 50 & 30\\
Max candidate actions per unit & 3.5 & 3.5\\
Number search iterations & 256 & 512\\
Policy sampling temperature for rollouts & 0.75 & N/A \\
Policy sampling top-p & 0.95 & 0.95\\
Rollout length, move phases & 2 & 0\\
\bottomrule
\end{tabular}
\caption{Search parameters used in \citet{gray2020human} compared to the search parameters used in our work. All experiments in the main body of this work uses search settings that are much cheaper to run. We use a rollout length of 0 and $N_c = 30$ while increasing the number of search iterations to 512.}
\label{tab:search_params_cmp}
\end{table*}

\begin{table*}[h!]
\small
\resizebox{\textwidth}{!}{
\begin{tabular}{lrrrrrr}
\toprule
\bf 1x \bf $\downarrow$ 6x \bf $\rightarrow$ & \bf DipNet & \bf DipNet RL &\bf  Blueprint &\bf  BRBot & \bf SearchBot \\
\midrule
DipNet \citep{paquette2019no} & - & 6.7\% $\pm$ 0.9\% & 11.6\% $\pm$ 0.1\% & 0.1\% $\pm$ 0.1\% & 0.7\% $\pm$ 0.2\% \\
DipNet RL \citep{paquette2019no} & 18.9\% $\pm$ 1.4\% & - & 10.5\% $\pm$ 1.1\% & 0.1\% $\pm$ 0.1\% & 0.6\% $\pm$ 0.2\%\\
\midrule
Blueprint \citep{gray2020human} & 20.2\% $\pm$ 1.3\% & 7.5\% $\pm$ 1.0\% & - & 0.3\% $\pm$ 0.1\% & 0.9\% $\pm$ 0.2\%\\
BRBot \citep{gray2020human} & 67.3\% $\pm$ 1.0\% & 43.7\% $\pm$ 1.0\% & 69.3\% $\pm$ 1.7\% & - & 11.1\% $\pm$ 1.1\%\\
SearchBot \citep{gray2020human} & 51.1\% $\pm$ 1.9\% & 35.2\% $\pm$ 1.8\% & 52.7\% $\pm$ 1.3\% & 17.2\% $\pm$ 1.3\% & -\\
\midrule
piKL-HedgeBot ($\lambda=0.001$) & 54.8\% $\pm$ 1.8\% & 31.4\% $\pm$ 1.8\% & 50.3\% $\pm$ 1.8\% & 19.2\% $\pm$ 1.4\% & 16.6\% $\pm$ 1.3\%\\
\midrule
piKL-HedgeBot ($\lambda=0.001$) (\cite{gray2020human} parameters) & 60.1\% $\pm$ 1.8\% & 33.3\% $\pm$ 1.8\% & 58.1\% $\pm$ 1.8\% & 23.6\% $\pm$ 1.6\% & 20.3\% $\pm$ 1.4\% \\
\bottomrule
\end{tabular}
}
\vspace{-2mm}
\caption{Average SoS scores achieved by the 1x agent against the 6x agents.
This table compares the performance of SearchBot \cite{gray2020human} and other agents from prior work with piKL-HedgeBot ($\lambda = 10^{-3}$) that uses a much cheaper search setting. Using the much cheaper search setting, comes at relatively small cost in its performance as we use improved value and policy models (See, \cref{appendix:diplomacyarchitecture} and \cref{appendix:hyperparams} for more details). When using the same parameters as \citet{gray2020human}, piKL-HedgeBot ($\lambda = 10^{-3}$) significantly outperforms SearchBot under most settings. Note that equal performance would be 1/7 $\approx$ 14.3\%. The $\pm$ shows one standard error.}
\label{tab:old_compare_agents}
\end{table*}

\section{Diplomacy Model Architecture and Input Features}
\label{appendix:diplomacyarchitecture}
Our imitation learning policy model for Diplomacy uses the same transformer-encoder LSTM-decoder architecture as \citet{bakhtin2021no} for reinforcement learning in Diplomacy, but applied to imitation learning over human games. This architecture also resembles the architecture used by a significant amount of past work \cite{gray2020human, anthony2020learning, paquette2019no} but replaces the graph-convolution encoder with a transformer, which we find to produce good results. 

Additionally, we slightly modify the input feature encoding relative to \citet{gray2020human}, removing a small number of redundant channels and adding channels to indicate the ``home centers'' of each of the 7 powers (Austria, England, France,... etc), which are the locations where that power is allowed to build new armies or fleets. See Table \ref{table:diplomacy-features} for the new list of input features. By adding the home centers to the input encoding instead of leaving them implicit, the game becomes entirely equivariant to permutations of those powers - e.g. if one swaps all the units of England and France, and all their centers, and which centers are their home centers, the resulting game is isomorphic to the original except with the two powers renamed. 

This allows us to then augment the training data via equivariant permutations of the seven possible powers in the encoding. Every time we sample a position from the dataset for training, we also choose among all 7-factorial permutations of the powers uniformly at random, and correspondingly permute both the input and output, to reduce overfit and improve the model's generalization given the limited human data available.

\begin{table*}[h]
\center
\begin{tabular}{llr}
\toprule
\bf Feature & \bf Type & \bf Number of Channels \\
\midrule
Location has unit? & One-hot (army/fleet), or all zero & 2 \\
Owner of unit & One-hot (7 powers), or all zero & 7 \\
Buildable, Removable? & Binary & 2 \\
Location has dislodged unit? & One-hot (army/fleet), or all zero & 2 \\
Owner of dislodged unit & One-hot (7 powers), or all zero & 7 \\
Area type & One-hot (land,coast,water) & 3 \\
Supply center owner & One-hot (7 powers or neutral), or all zero & 8 \\
Home center & One-hot (7 powers), or all zero & 7 \\
\bottomrule
\end{tabular}
\caption{Per-location input features used}
\label{table:diplomacy-features}
\end{table*}

\section{Improved Value Model in Diplomacy}
\label{appendix:propval}
For no-press Diplomacy, we note that \citet{gray2020human} observed that their search agent benefits from short rollouts using the trained human policy before applying the human-learned value model to evaluate the position. Doing so appears to result in more accurate evaluations reflecting the likely outcomes from a given game state, which the raw value model may failed to learn sufficiently accurately on the limited human dataset. Since expectation of the learned value model after a short rollout appears to be better than the learned value model itself, this motivates training a model to directly approximate the former.

In a fashion broadly similar to \citet{silver2016mastering} generating rollout games to train a more accurate value head for Go, we therefore generated a large stream of data by uniformly sampling positions from the human game dataset for Diplomacy, rolling them forward between 4-8 phases of game play via the same rollout settings as \citet{gray2020human}, i.e. policy sampling temperature 0.75, top-p 0.95, and training a new value model to predict the resulting post-rollout value estimate of the old value model. Samples were continuously and asynchronously added to a replay buffer of 10000 batches, and the buffer was continuously sampled to train the same transformer-based architecture as the human-trained model from Appendix \ref{appendix:diplomacyarchitecture} initialized with the weights of that model. Training was constrained to never exceed the rate of data generation by more than a factor of 2 (i.e. using each sample twice in expectation) and proceeded for 128000 mini-batches of 1024 samples each using the ADAM optimizer with a fixed learning rate of 1e-5.

\section{More Experiments in Diplomacy}
\label{appendix:population_experiments}

This section compiles additional performance results from evaluation games.

\subsection{Head-to-Head Performance}
In Table \ref{tab:1v6-piKL}, we compare the performance of \qre in 1v6 head-to-head games against the underlying imitation anchor policy, following prior work \cite{gray2020human, bakhtin2021no, paquette2019no, anthony2020learning}. We find that the $\lambda=10^{-1}$ policy is substantially stronger than the imitation policy while matching the accuracy in predicting human moves, while the $\lambda=10^{-3}$ policy outperforms unregularized search methods while playing much closer to the human policy.

\begin{table*}[ht!]
\centering
\small
\resizebox{\textwidth}{!}{
\begin{tabular}{llr}
\toprule
 1x & 6x & \textbf{Average SoS Score} \\
\midrule
\multirow{7}{*}{\bp} & piKL-HedgeBot ($\lambda=10^{-1}$)       &      8.3$\pm$0.9\%    \\    
& piKL-HedgeBot ($\lambda=10^{-2}$)     &   2.5$\pm$0.4\%    \\   
& piKL-HedgeBot ($\lambda=10^{-3}$)     &   1.8$\pm$0.3\%    \\   
& piKL-HedgeBot ($\lambda=10^{-4}$)     &   2.1$\pm$0.3\%    \\   
& piKL-HedgeBot ($\lambda=10^{-5}$)     &   1.6$\pm$0.2\%    \\  
& HedgeBot     &   1.5$\pm$0.2\%    \\  
& \sbot           &     1.4$\pm$0.2\%    \\    
\bottomrule
\end{tabular}
\quad
\begin{tabular}{llr}
\toprule
\bf 1x & \bf 6x & \bf \textbf{Average SoS Score} \\
\midrule
piKL-HedgeBot ($\lambda=10^{-1}$)      & \multirow{7}{*}{\bp} &  21.1$\pm$1.4\%     \\   
piKL-HedgeBot ($\lambda=10^{-2}$)     &  &  44.2$\pm$1.7\%    \\      
piKL-HedgeBot ($\lambda=10^{-3}$)    &  &  52.7$\pm$1.7\%     \\  
piKL-HedgeBot ($\lambda=10^{-4}$)     &  &  49.7$\pm$1.7\%    \\      
piKL-HedgeBot ($\lambda=10^{-5}$)    &  &  46.9$\pm$1.7\%     \\
HedgeBot                   &  &  46.5$\pm$1.7\%     \\
\sbot                     &  &  46.2$\pm$1.7\%    \\  
\bottomrule
\end{tabular}
}
\caption{Average SoS score attained by the 1x agent against the 6x agent. \qre ($\lambda=10^{-1}$) policy is substantially stronger than IL Policy, while the ($\lambda=10^{-2}$) policy is almost as strong as \sbot. The $\pm$ shows one standard error. Note that equal performance would be 1/7 $\approx$ 14.3\%. 
}
\label{tab:1v6-piKL}
\end{table*}

\subsection{piKL-HedgeBot's performance in population-based experiments}

In this section, we provide all the results from the population experiments across various piKL-HedgeBot's lambda values. \cref{fig:diplomacypopulation} and \cref{tab:diplomacypopulation} show the results. piKL-HedgeBot with $\lambda=10^{-3}$ performs best across individual population experiments with an SoS score of 32.9\%. The performance drops as we continue to increase $\lambda$ past $1e-3$. Error bars indicate 1 standard error.

\begin{figure*}[ht!]
\centering
\includegraphics[width=15cm]{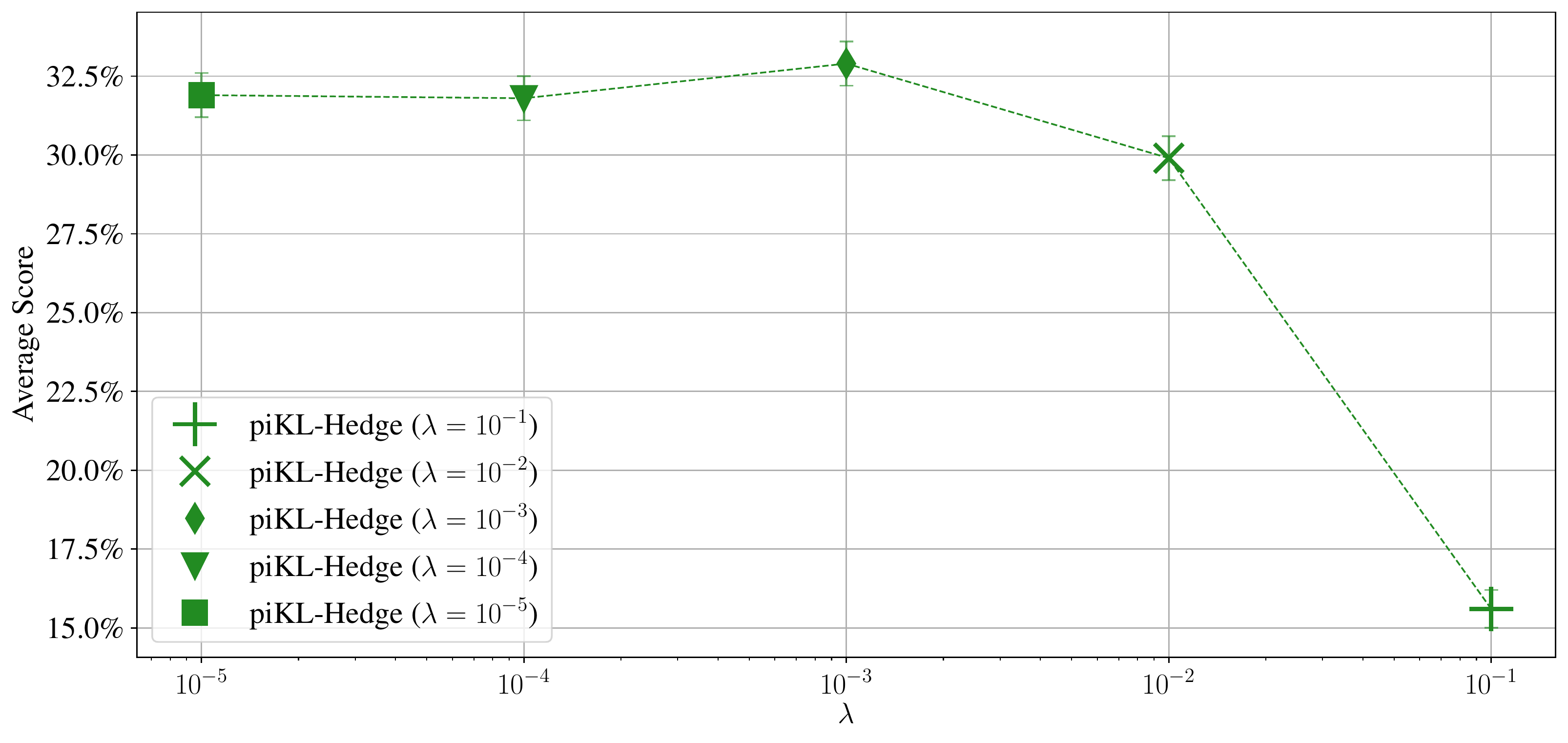}
\vspace{-0.2in}
\caption{\small Average SoS score achieved by piKL-HedgeBots in uniformly sampled pools of other agents as a function of $\lambda$. piKL-HedgeBot ($\lambda = 10^{-3}$) performs best across individual sweeps with an SoS score of 32.9\%.}
\label{fig:diplomacypopulation}
\end{figure*}

\begin{table*}
\centering
\small
\begin{tabular}{lr}
\toprule
\bf \textbf{Agent} & \bf \textbf{Average SoS Score}\\
\midrule
DipNet \citep{paquette2019no}    & 4.9\% $\pm$ 0.3\% \\
DipNet RL \citep{paquette2019no}   & 5.6\% $\pm$ 0.4\% \\
\midrule
Blueprint  \citep{gray2020human}   & 7.1\% $\pm$ 0.4\% \\
BRBot \citep{gray2020human}       & 18.2\% $\pm$ 0.6\% \\
SearchBot \citep{gray2020human}  & 36.1\% $\pm$ 0.8\% \\
\midrule
\bp   & 10.2\% $\pm$ 0.6\% \\
\sbot  & 36.8\% $\pm$ 1.1\% \\
\midrule
piKL-HedgeBot ($\lambda=10^{-1}$)      &      15.6\% $\pm$ 0.6\% \\
\bottomrule
\end{tabular}
\quad
\begin{tabular}{lr}
\toprule
 \textbf{Agent} & \textbf{Average SoS Score}\\
\midrule
DipNet \citep{paquette2019no}    & 3.8\% $\pm$ 0.3\% \\
DipNet RL \citep{paquette2019no}   & 4.2\% $\pm$ 0.3\% \\
\midrule
Blueprint  \citep{gray2020human}   & 5.8\% $\pm$ 0.4\% \\
BRBot \citep{gray2020human}       & 16.3\% $\pm$ 0.6\% \\
SearchBot \citep{gray2020human}  & 14.1\% $\pm$ 0.6\% \\
\midrule
\bp   & 8.5\% $\pm$ 0.4\% \\
\sbot  & 31.7\% $\pm$ 0.8\% \\
\midrule
piKL-HedgeBot ($\lambda=10^{-2}$)      &      29.9\% $\pm$ 0.7\% \\
\bottomrule
\end{tabular}
\quad
\begin{tabular}{lr}
\toprule
 \textbf{Agent} & \textbf{Average SoS Score}\\
\midrule
DipNet \citep{paquette2019no}    & 3.7\% $\pm$ 0.3\% \\
DipNet RL \citep{paquette2019no}   & 4.7\% $\pm$ 0.3\% \\
\midrule
Blueprint  \citep{gray2020human}   & 4.9\% $\pm$ 0.3\% \\
BRBot \citep{gray2020human}       & 16.1\% $\pm$ 0.6\% \\
SearchBot \citep{gray2020human}  & 13.4\% $\pm$ 0.5\% \\
\midrule
\bp   & 7.9\% $\pm$ 0.4\% \\
\sbot  & 31.3\% $\pm$ 0.7\% \\
\midrule
piKL-HedgeBot ($\lambda=10^{-3}$)      &      32.9\% $\pm$ 0.7\% \\
\bottomrule
\end{tabular}
\quad
\begin{tabular}{lr}
\toprule
 \textbf{Agent} & \textbf{Average SoS Score}\\
\midrule
DipNet \citep{paquette2019no}    & 3.5\% $\pm$ 0.3\% \\
DipNet RL \citep{paquette2019no}   & 4.6\% $\pm$ 0.3\% \\
\midrule
Blueprint  \citep{gray2020human}   & 5.7\% $\pm$ 0.3\% \\
BRBot \citep{gray2020human}       & 14.3\% $\pm$ 0.6\% \\
SearchBot \citep{gray2020human}  & 13.6\% $\pm$ 0.5\% \\
\midrule
\bp   & 8.8\% $\pm$ 0.4\% \\
\sbot  & 31.8\% $\pm$ 0.7\% \\
\midrule
piKL-HedgeBot ($\lambda=10^{-4}$)     &      31.8\% $\pm$ 0.7\% \\
\bottomrule
\end{tabular}
\quad
\begin{tabular}{lr}
\toprule
 \textbf{Agent} & \textbf{Average SoS Score}\\
\midrule
DipNet \citep{paquette2019no}    & 3.6\% $\pm$ 0.3\% \\
DipNet RL \citep{paquette2019no}   & 4.4\% $\pm$ 0.3\% \\
\midrule
Blueprint  \citep{gray2020human}   & 5.3\% $\pm$ 0.3\% \\
BRBot \citep{gray2020human}       & 15.0\% $\pm$ 0.6\% \\
SearchBot \citep{gray2020human}  & 13.0\% $\pm$ 0.5\% \\
\midrule
\bp   & 8.9\% $\pm$ 0.4\% \\
\sbot  & 32.2\% $\pm$ 0.7\% \\
\midrule
piKL-HedgeBot ($\lambda=10^{-5}$)      &      31.9\% $\pm$ 0.7\% \\
\bottomrule
\end{tabular}
\quad
\begin{tabular}{lr}
\toprule
 \textbf{Agent} & \textbf{Average SoS Score}\\
\midrule
DipNet \citep{paquette2019no}    & 3.6\% $\pm$ 0.3\% \\
DipNet RL \citep{paquette2019no}   & 3.9\% $\pm$ 0.3\% \\
\midrule
Blueprint  \citep{gray2020human}   & 5.5\% $\pm$ 0.3\% \\
BRBot \citep{gray2020human}       & 14.7\% $\pm$ 0.6\% \\
SearchBot \citep{gray2020human}  & 14.1\% $\pm$ 0.6\% \\
\midrule
\bp   & 8.7\% $\pm$ 0.4\% \\
\sbot  & 32.2\% $\pm$ 0.7\% \\
\midrule
HedgeBot     &      31.7\% $\pm$ 0.7\% \\
\bottomrule
\end{tabular}

\caption{Average SoS score achieved by agents in uniformly sampled pools of other agents. piKL-HedgeBot with $\lambda=10^{-3}$ performs best across individual sweeps with an SoS score of 32.9\%. The $\pm$ shows one standard error.}
\label{tab:diplomacypopulation}
\end{table*}

\clearpage
\section{Dec-POMDP Games: Policy-regularized SPARTA on Hanabi}
\label{appendix:hanabi}

In this section, we extend KL-regularized search to decentralized partially observable Markov decision processes (Dec-POMDP) and test it on the Hanabi benchmark~\cite{bard2020hanabi}. We first train an imitation learning policy (IL policy) from human data. We then use the IL policy as the blueprint policy in SPARTA~\cite{lerer2020improving}, a search technique for Dec-POMDPs, and apply KL-regularization toward the IL policy in SPARTA. We call this new algorithm piKL-SPARTA. We show that piKL-SPARTA matches or even slightly improves the original IL policy in human move prediction accuracy while greatly improving self-play performance.

\subsection{Background}

A Dec-POMDP is an $N$-player fully cooperative game with state space $S$ that is partially observed by each player $i$ through their individual observation function $o_i = \Omega_i(s)$ for $s \in S$, with joint action space $A = A_1 \times A_2 \times \cdots \times A_N$ and transition function $\mathcal{T} \colon S \times A \to P(S)$ that returns the distribution of next state given current state and joint action. 
The reward function $R \colon S \times A \to \mathbb{R}$ assigns a scalar reward for the \textit{entire} team at each time step. 
A trajectory is denoted as $\tau^t = (s^0, \textbf{a}^0, r^0, \ldots, s^t)$ while the action-observation history (AOH) of each player is defined as $\tau^t_i = (o^0, \textbf{a}^0, r^0, \ldots, o^t)$. A full trajectory or full AOH that reaches the terminal state may be denoted more simply as $\tau$ or $\tau_i$ respectively. The policy for each individual player $\pi_i(a_i^t | \tau^t_i)$ takes as input the AOH and returns a distribution over valid actions. The joint policy $\pi = (\pi_1, \ldots, \pi_N)$ is a tuple containing all players' policies. The goal is to find a policy to maximize the expected total return $\pi^{*} = \argmax_{\pi} J(\pi) =  \mathbb{E}_{\tau \sim P(\tau | \pi)} R^0(\tau)$ where $R^t(\tau) = \sum_{t' \geq t} \gamma^{(t' - t)} r_t'$ is the forward looking return with optional discount factor $\gamma \leq 1$.

Hanabi is a well-established large-scale Dec-POMDP benchmark~\cite{bard2020hanabi}. It is a 2 to 5 player card game with a deck of 50 cards equally divided into 5 color suits. Each color consists of five ranks with three 1s, two 2s, two 3s, two 4s, and one 5. Each player draws five cards from a randomly shuffled deck to start the game. The goal of the team is to play cards in order of increasing rank from 1 to 5 for every color suit. Players take turns to either play a card, discard a card, or give a hint to another player about their cards. When giving a hint, the acting player picks a receiver of the hint and a rank or color of any card in the receiver's hand. The recipient will then learn exactly which cards in their hand match the given rank or color. Hinting costs one information token. The team starts with eight information tokens and they can recoup one information token after discarding a card or successfully playing a 5 of any color. If a player makes an invalid play, e.g. playing a red 3 when a red 2 is not played yet, the team loses one life token. After each play or a discard, a player draws a new card if possible. The game ends when three life tokens are lost, in which case the team receives 0 points, or one round after the entire deck is exhausted, in which case the score equals the number of cards successfully played. 

The majority of existing works in the Hanabi domain focus on learning human compatible policies without using human data~\cite{bard2020hanabi, siu2021, hu2021off}. Fewer works have explored better modeling human policies using human game data, partly due to the lack of publicly available datasets. To the best of our knowledge, the strongest supervised learning agent trained from human data was done by~\cite{hu2021off} where the authors use it as an unseen test-time partner agent in evaluation to estimate how well their agents might collaborate with humans. No prior work has been done to better predict human moves in Hanabi.

A few search techniques such as SPARTA~\cite{lerer2020improving} and RL-search~\cite{rlsearch} have been proposed for large scale Dec-POMDPs and have specifically been applied in Hanabi. In this work we choose SPARTA as our backbone for simplicity, while noting that our KL-regularization methods may also be generalized to other search techniques. SPARTA is a test-time policy improvement algorithm that can be applied on top of any policy. SPARTA assumes that a \textbf{blueprint} policy (BP) $\pi$ is common knowledge in the Dec-POMDP and all players play their part of the blueprint unless they should deviate according to the SPARTA rule. Here we briefly discuss the single-agent variant where the search agent assumes that other agents will always play the blueprint. Given blueprint $\pi$, SPARTA first defines a belief function that tracks the distribution of the real trajectory given the search agent's own AOH $\mathcal{B}_i(\tau^t) = P(\tau^t | \tau_i^t, \pi)$. Then the search agent $i$ computes the expected value for each action $a$ using Monte Carlo rollouts: 
\begin{align}
Q_{\pi}(\tau_i^t, a) &= \mathbb{E}_{\tau^t \sim \mathcal{B}_i(\tau^t)} Q_{\pi}(\tau^t, a), \label{eq:sparta-q}
\end{align}
where:
\begin{align*}
Q_{\pi}(\tau^t, a) &= \mathbb{E}_{\tau \sim P(\tau |\mathcal{T}, \tau^t, a_i^t = a, a^t_{j\neq i} \sim \pi, \textbf{a}^{t' > t} \sim \pi)} R^t(\tau)
\end{align*}
is the expected forward looking return on $\tau^t$ assuming that the search agent will perform the action $a$ for the current step and follow the blueprint afterwards while other agents always follow the blueprint. SPARTA computes the belief $\mathcal{B}_i(\tau_t)$ analytically. It maintains a distribution of all possible trajectories and adjusts it at every step by removing the trajectories that contradict public knowledge or would have led to different joint actions according to the known joint policy $\pi$.

\subsection{Method}
We start with an imitation learning policy (IL policy) trained from human gameplay data collected from an online Hanabi game platform. The IL policy is used as both the baseline for predicting human moves as well as the blueprint policy for our piKL-SPARTA. 
The analytical belief update procedure in SPARTA requires full knowledge of the partners' policy. This is challenging when predicting human moves and playing with humans because our IL policy models the average behavior of the entire population of players in the training dataset and a player at test time may perform actions that the IL policy will never do, leading to a null belief space that will terminate the search. Therefore, we follow the practice in~\cite{lbs} and train an approximate neural network belief model on self-play data generated by the IL policy. 

Given the IL policy $\pi$ as blueprint and the approximate belief model $\hat{\mathcal{B}}_i$ that replaces the $\mathcal{B}_i$ in Eq.~\ref{eq:sparta-q}, piKL-SPARTA selects actions for the search player $i$ following
\begin{align}
    P(a) \propto \pi(a | \tau_i^t) \cdot \exp \left[\frac{Q_{\pi}(\tau_i^t, a)}{\lambda}\right].
    \label{eq:pikl-sparta}
\end{align}

\subsection{Experimental Setup}
We use a similar dataset acquired from \url{en.boardgamearena.com} as in~\cite{hu2021off}. The dataset consists of 240,954 2-player Hanabi games. We randomly sample 1,000 games to create a validation set and another 4,000 games for the test set. The training set contains the remaining 235,954 games with an average score of 15.88. Each game records the AOH $\tau_i, i\in \{1, 2\}$ for both players. The IL policy $\pi_\theta$ is parameterized by a neural network $\theta$ and is trained to minimize the cross-entropy loss 
$$ \mathcal{L(\theta)} = - \mathbb{E}_{\tau_i \sim D}\sum_{t=0}^{T} \pi_\theta(a_i^t|\tau_i^t) $$
with stochastic gradient descent. The training set $D$ is dynamically augmented with color shuffling~\cite{hu2020other} where a random color permutation that changes both observation and action space is applied to each trajectory $\tau_i$ sampled from the dataset before feeding it to the network. \citet{hu2021off} shows that this data augmentation method greatly reduces overfitting and leads to better policies. The network $\theta$ uses the Public-LSTM structure in~\cite{lbs}, which eliminates the need to re-unroll LSTM on sampled trajectories from the beginning of the game as they share the same public observations as the real trajectory. 

To construct the approximate belief model $\hat{\mathcal{B}}$, we train a neural network $\phi$ to predict each player's own hand, which is the only hidden information in Hanabi. 
The network predicts each card in hand from oldest to newest auto-regressively. We denote the hidden cards as $\{h_i^{j}\}$ with $h^1_i$ being the oldest and $h^m_i$ being the newest, $m \leq 5$. Then the cross-entropy loss for $\phi$ in an $N$-player setting becomes
\begin{align*}
    \mathcal{L(\phi)} = -\mathbb{E}_{\tau \sim \pi_\theta}\left[ \frac{1}{N}\sum_{i=1}^{N} \sum_{t=1}^{T} \sum_{j=1}^{m} \log P_{\phi}(h^j_i | \tau_i^t, h^1_i, \ldots, h_i^{j-1}) \right].
\end{align*}
In Hanabi, it is sufficient to reconstruct $\tau^t$ given $\tau^t_i$ and $\{h^i_j\}$.
The model is trained on infinite stream of data generated by $\pi_\theta$ through self-play to avoid overfitting. We train two variants of the belief model, one with sampled action $\textbf{a} \sim \pi_\theta$ while the other with greedy action $\textbf{a} = \argmax \pi_\theta$. They are used by piKL-SPARTA and piKL-SPARTA-G(reedy) respectively.

We first run piKL-SPARTA on the test set to compare its ability to predict human moves against the IL policy. For each $\tau_i^t$ in the test set, we sample $\frac{K}{|A_i|}$ hands from the belief model $P_{\phi}$ where $K$ is the total number of searches for this step and ${|A_i|}$ is the number of legal actions of the search player. In all our experiments we set $K=10,000$. In practice, we sample $\frac{2K}{|A_i|}$ hands and take the top $\frac{K}{|A_i|}$ samples not contradicting the public knowledge. If the belief model fails to produce any samples that comply with the public knowledge, we revert back to the blueprint. We then compute the expected value for each search action by unrolling the IL policy until the end of the game for each sampled trajectory and compare $a_{\text{pred}} = \argmax_a \pi_{\theta}(a | \tau_i^t) \cdot \exp [\frac{Q_{\pi}(\tau_i^t, a)}{\lambda}]$ against the human move $a^t_i$ from the data. We experiment with different $\lambda$ to study its effect on prediction accuracy.

We also evaluate piKL-SPARTA in self-play to compare its performance under different $\lambda$ against the IL policy. We run piKL-SPARTA and the IL policy on 4,000 games with different seeds for the deck. At each step, the IL policy acts following $a_t \sim \pi_\theta(a | \tau^t_i)$ while piKL-SPARTA acts following Eq.~(\ref{eq:pikl-sparta}) again using $K=10,000$ rollouts.

The $\lambda$ in these experiments are significantly higher than those in Diplomacy or implicitly in MCTS in Chess and Go\footnote{Via the relationship $\lambda \approx c_{\text{puct}} \sqrt{N}$ where $N$ is the number of MCTS iterations} because $\lambda$ represents the scale of utility difference that offsets a particular KL penalty, and the range of the utilities $Q_{\pi}(\tau_i^t, a)$ in Hanabi is $[-25, 25]$ whereas the range in other games are either $[0, 1]$ or $[-1, 1]$. 

We experiment with both 1p piKL-SPARTA, where one player uses piKL-SPARTA and the other follows the blueprint, and 2p piKL-SPARTA, where both players run the same single-agent version of piKL-SPARTA independently. The latter is not theoretically sound and 2p SPARTA has in past papers produced worse performance than 1p SPARTA~\cite{lerer2020improving} because 1p SPARTA assumes that partners are playing according to the blueprint policy while in actuality the partners are playing according to a SPARTA policy. Despite the lack of theoretical soundness, we are interested in whether 2p piKL-SPARTA can obtain empirical improvement over the 1p version, since piKL-SPARTA regularizes towards the blueprint and therefore the mismatch between assuming the partner follows the blueprint and the policy they actually play would likely be less severe. 

Lastly, in Dec-POMDPs, it is common to select actions greedily in self-play instead of sampling from a mixed policy, since in fully cooperative settings there is no need to avoid being deterministic or predictable to an adversary. Therefore, we also experimented with piKL-SPARTA-G(reedy) where every agent, including the baseline IL policy, play according to the $\argmax$ of their action distributions at every step, and the belief model for piKL-SPARTA-G is trained on trajectories produced by a greedy IL policy.

\subsection{Results}
\begin{figure}[ht!]
    \centering
    \includegraphics[width=10cm]{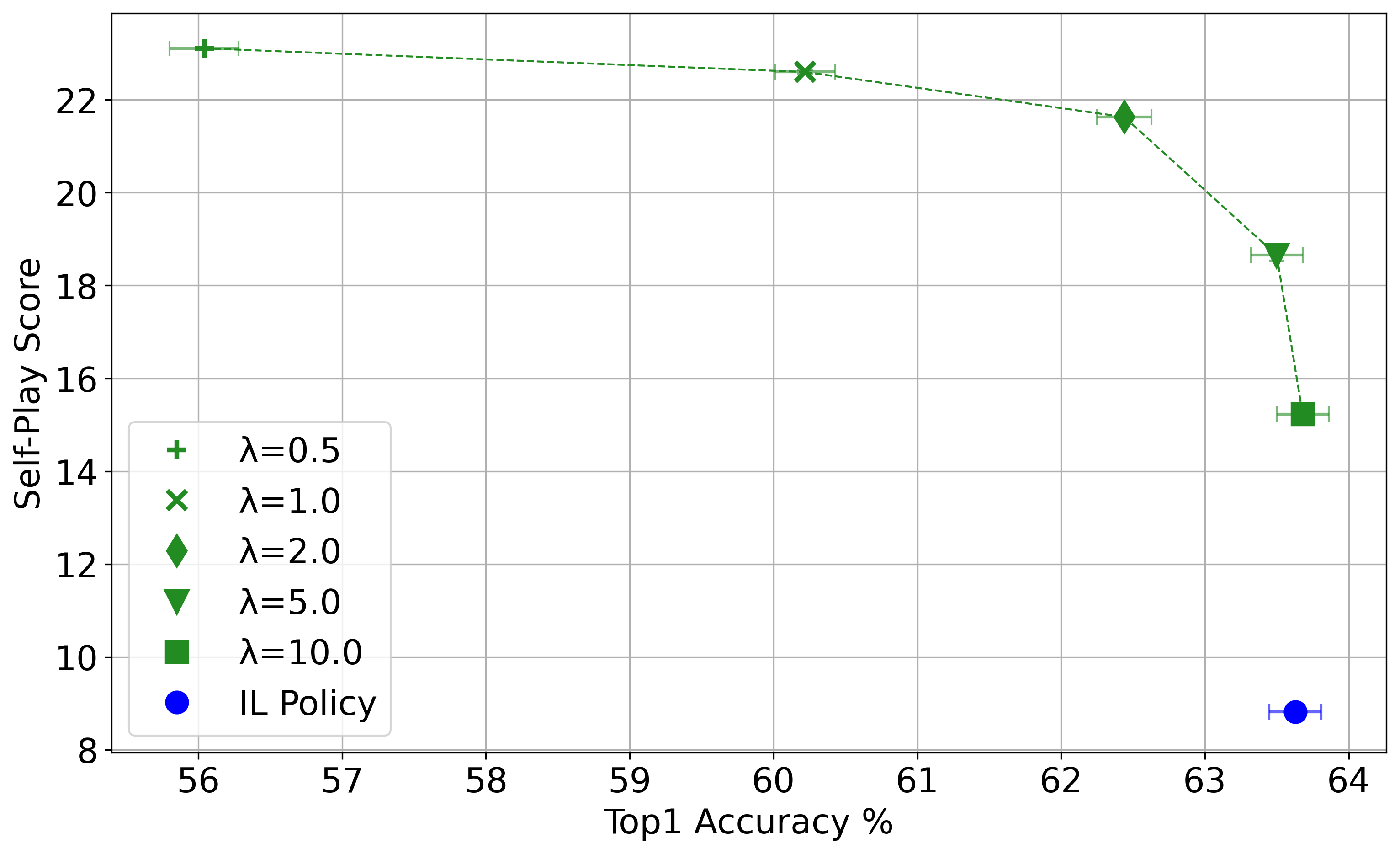}
    \caption{Top-1 test accuracy and self-play score of IL policy and piKL-SPARTA in Hanabi. Blue dot is the IL policy and green dots are piKL-SPARTA with different $\lambda$. The self-play score is evaluated with sampling based 2p piKL-SPARTA and sampling based SL policy. Additional evaluations with more $\lambda$ and more algorithm variants are presented in Table~\ref{tab:hanabi-pred} and Table~\ref{tab:hanabi-score}. Error bar is 1 standard error.}
    \label{fig:hanabi-result}
\end{figure}

\begin{table}[]
    \small
    \centering
    \begin{tabular}{ c c c c c | c}
    \toprule
    Subset of Test Set & $\lambda=$0 & $\lambda=$0.5 & $\lambda=$1 &$\lambda=$2 & IL Policy \\
    \midrule
    All games  &  25.30\% $\pm$ 0.16\%  &  56.04\% $\pm$ 0.24\%  & 60.22\% $\pm$ 0.21\% & 62.44\% $\pm$ 0.19\% & 63.63\% $\pm$ 0.18\% \\
    Games w/score $\geq$ 10 & 27.01\% $\pm$ 0.18\% & 58.86\%  $\pm$ 0.25\% & 62.62\% $\pm$ 0.22\% & 64.51\% $\pm$ 0.20\% & 65.29\% $\pm$ 0.19\% \\
    Games w/score $\geq$ 20 & 28.87\% $\pm$ 0.19\% & 61.86\% $\pm$ 0.25\% & 65.21\% $\pm$ 0.23\% & 66.81\% $\pm$ 0.21\% & 67.39\% $\pm$ 0.19\% \\
    \bottomrule
    \toprule
    Subset of Test Set & $\lambda=$5 & $\lambda=$10 & $\lambda=$20 & $\lambda=$50 & IL Policy \\
    \midrule
    All games  &  63.50\% $\pm$ 0.18\% & 63.68\% $\pm$ 0.18\% & \textbf{63.71\% $\pm$ 0.18\%} & 63.68\% $\pm$ 0.18\% & 63.63\% $\pm$ 0.18\% \\
    Games w/score $\geq$ 10 & 65.33\% $\pm$ 0.19\% & \textbf{65.43\% $\pm$ 0.19\%} & 65.42\% $\pm$ 0.19\% & 65.35\% $\pm$ 0.19\% & 65.29\% $\pm$ 0.19\% \\
    Games w/score $\geq$ 20 & 67.48\% $\pm$ 0.19\% & \textbf{67.54\% $\pm$ 0.19\%} & 67.53\% $\pm$ 0.19\% & 67.45\% $\pm$ 0.19\% & 67.39\% $\pm$ 0.19\% \\
    \bottomrule
    \end{tabular}
    \caption{Human prediction accuracy of unregularized SPARTA ($\lambda = 0$) and of piKL-SPARTA with different $\lambda$ and the IL policy on test set. Each row represents their accuracy on a subset filtered by the final score of the games. pikl-SPARTA achieves similar prediction accuracy as IL for most $\lambda$ and is far more accurate than unregularized SPARTA.}
    \label{tab:hanabi-pred}
\end{table}

\begin{table}[]
    \setlength{\tabcolsep}{1.2mm}
    \small
    \centering
    \begin{tabular}{ c c c c c c c | c}
    \toprule
    & $\lambda=$0 & $\lambda=$0.5 & $\lambda=$1 & $\lambda=$2 & $\lambda=$5 & $\lambda=$10 & IL Policy \\
    \midrule
    1p piKL-SPARTA & 20.56 $\pm$ 0.05 & 21.16 $\pm$ 0.06 & 20.02 $\pm$ 0.09 & 17.71 $\pm$ 0.12 & 13.53 $\pm$ 0.16 & 11.04 $\pm$ 0.17 & \multirow{2}{*}{8.81 $\pm$ 0.15} \\
    2p piKL-SPARTA & 20.23 $\pm$ 0.04 & \textbf{23.11 $\pm$ 0.03} & 22.60 $\pm$ 0.03 & 21.63 $\pm$ 0.05 & 18.65 $\pm$ 0.11 & 15.23 $\pm$ 0.15 & \\
    \midrule
    1p piKL-SPARTA-G & 22.78 $\pm$ 0.03 & 23.07 $\pm$ 0.03 & 22.76 $\pm$ 0.03 & 22.41 $\pm$ 0.04 & 21.91 $\pm$ 0.05 & 21.45 $\pm$ 0.06 & \multirow{2}{*}{19.72 $\pm$ 0.10} \\
    2p piKL-SPARTA-G & 19.98 $\pm$ 0.04 & \textbf{23.72 $\pm$ 0.02} & 23.39 $\pm$ 0.03 & 22.93 $\pm$ 0.03 & 22.36 $\pm$ 0.03 & 21.98 $\pm$ 0.04 &  \\
    \bottomrule
    \end{tabular}
    \caption{Performance of unregularized SPARTA ($\lambda = 0$) and piKL-SPARTA under different $\lambda$ and IL policy evaluated on 4,000 self-play games, reported $\pm$ one standard error. In the top two rows, both piKL-SPARTA and IL policy sample actions according to their action distribution respectively. In the bottom two rows, both algorithms take greedy actions; ``-G" is short for ``-Greedy". All numbers shown in each table section use the same learned belief model for fair comparison. At the high lambdas that maintain or improve human accuracy, piKL-SPARTA significantly improves playing strength over IL, while at lower lambda values piKL-SPARTA outscores unregularized SPARTA. For reference, unregularized greedy 1p SPARTA with exact beliefs rather than learned beliefs gets 23.49 $\pm$ 0.02.}
    \label{tab:hanabi-score}
\end{table}

Table~\ref{tab:hanabi-pred} summarizes the results for human move prediction. On the full test set, piKL-SPARTA with $\lambda = 10$ and $\lambda = 20$ outperform the IL policy slightly, although not statistically confidently, and vastly outperform unregularized SPARTA. We also investigate the prediction accuracy on games with score $\geq 10$ or $\geq 20$ that more predominantly come from more experienced human players. The improvement of piKL-SPARTA over the IL policy may be larger on these games, although this result is noisy and at best should be considered only mildly suggestive since filtering on final score also introduces other major confounding factors as well.

In Table~\ref{tab:hanabi-score}, we show the self-play performance piKL-SPARTA and IL policy. The top two rows show the results of the sampling version while the bottom two rows show those of the greedy version. The conclusions are consistent across both cases. Both 1p and 2p variants of piKL-SPARTA outscore IL for all $\lambda$ tested. Together with Table~\ref{tab:hanabi-pred} we find with $\lambda = 10$, piKL-SPARTA maintains IL prediction accuracy and outscores IL greatly in self-play, an overall improvement without any tradeoff. For smaller $\lambda = 2$ or $\lambda = 5$, piKL-SPARTA further outscores IL while losing some prediction accuracy on human moves, but remains vastly more accurate than unregularized SPARTA ($\lambda = 0$). 

Through qualitative analysis of the games played by both methods, we find that IL makes many mistakes due to both sampling low probability actions and the fact that the training set contains bad moves. piKL-SPARTA, on the other hand, avoids many of these mistakes while still otherwise following the same strategies and humanlike signaling conventions. It is particularly good at preventing catastrophic failures where the agents lose all life tokens and points since the $Q$ values for good and bad actions differ much more widely in those cases. 

The 2p piKL-SPARTA variants, despite being theoretically unsound, result in a further score improvement. Meanwhile, we notice that 2p versions of the plain SPARTA ($\lambda = 0$) underperform their 1p variants, which is consistent with the observations from~\cite{lerer2020improving} despite us using an approximate learned belief model instead of exact belief. This suggests that regularization towards the blueprint IL policy is successful in keeping  the trajectories similar enough to those from the blueprint that the unsound assumption that the partner plays the blueprint does not cause major problems, while still allowing both players to deviate enough to correct major blunders that they would otherwise make.

Lastly, we notice that 1p piKL-SPARTA with $\lambda=0.5$ outperforms the unregularized version in self-play. The reason could be that the quality of the samples from the learned belief model worsen as the trajectories it sees at test time becomes more off-distribution for smaller $\lambda$. For reference, we also run the original SPARTA which does not have this problem as it computes beliefs analytically. The original SPARTA gets 23.49 $\pm$ 0.02, which is better than 1p piKL-SPARTA with $\lambda=0.5$ but worse than 2p piKL-SPARTA with the same $\lambda$. The computational cost of 1p SPARTA with exact beliefs and 2p piKL-SPARTA with approximate learned beliefs is roughly the same because the exact beliefs computation accounts for 50\% of the entire computation. Therefore, 2p piKL-SPARTA could be a preferable option to improve performance in Dec-POMDPs via self play without incurring the huge cost of joint SPARTA or RL-search~\cite{rlsearch}.

\end{document}